\documentclass{IEEEojcsys}

\usepackage[colorlinks,urlcolor=blue,linkcolor=blue,citecolor=blue]{hyperref}

\usepackage{color,array}

\usepackage{graphicx}
\usepackage{psfrag,cite}
\usepackage{tabularx}
\usepackage{amssymb,amsmath,epsfig}
\usepackage{mathrsfs}
\usepackage{subcaption}

\newcommand{\expvalue}{\mathbb{E}} 
\newcommand{\XC}{X} 
\newcommand{\XE}{X_e} 
\newcommand{\XCH}{\hat{X}} 
\newcommand{\XEH}{\hat{X}_e}
\newcommand{\real}{\mathbb{R}}
\newcommand{\integerspos}{\mathbb{N}}

\newcommand{\timenN}{t_n}

\newcommand{\timepS}{t'_n}

\newcommand{\gmm}{1+\nu}

\newcommand{\ngmm}{1-\nu}

\newcommand{\oprocend}{\relax\ifmmode\else\unskip\hfill\fi\oprocendsymbol}

\newcommand{\oprocendsymbol}{\hbox{$\bullet$}}

\newtheorem{remark}{Remark}
\newtheorem{defn}{Definition}

\newtheorem{corollary}{Corollary}

\newtheorem{theorem}{Theorem}
\newtheorem{lemma}{Lemma}
\begin{document}

\sptitle{Article Category}

\title{Stabilizing a linear system using phone calls     when   time is information} 

%\editor{This paper was recommended by Associate Editor F. A. Author.}

\author{Mohammad Javad Khojasteh, \affilmark{1} (Member, IEEE)}

\author{Massimo Franceschetti\affilmark{2}  (Fellow, IEEE)}

\author{Gireeja Ranade\affilmark{3}}

\affil{EME department at Rochester Institute of Technology. Portions of this work were completed  at University of California, San Diego, Massachusetts Institute of Technology, and while visiting Microsoft Research.} 
\affil{ECE department at the University of California, San Diego.} 
\affil{EECS department at the University of California, Berkeley. Some of this work was performed while at Microsoft Research, and also at the Simons Institute for the Theory of Computing}

\corresp{CORRESPONDING AUTHOR: MJ Khojasteh (e-mail: \href{mjkeme@rit.edu}{mjkeme@rit.edu})}
\authornote{This research was partially supported by NSF awards CNS-1446891 and ECCS-1917177, and the Gleason Endowment at RIT.}

\begin{abstract}
We study stabilizing a scalar linear system over a ``timing'' channel, where information is communicated via the timestamps of transmitted symbols. Each symbol, sent from a sensor to a controller in a closed-loop system, is received after a random delay. The sensor encodes information in the waiting times between transmissions, and the controller decodes it using the inter-reception times of symbols. This setup resembles a telephone system: a transmitter signals a call with a ``ring,'' and the receiver becomes aware of the ``ring'' after a random connection delay. With no data payload exchanged, this setup provides an abstraction for performing event-triggering control with zero-payload rate. We show that for system stabilization (state converging to zero in probability), the timing capacity of the channel must be at least as large as the entropy rate of the system. For exponentially distributed symbol delays, we provide an ``almost'' tight sufficient condition using a coding strategy that refines the message estimate with each new symbol. Our results generalize previous zero-payload event-triggering control strategies, revealing a fundamental limit for using timing information in stabilization, independent of the transmission strategy.
\end{abstract}

\begin{IEEEkeywords}
Timing channel,  control with communication constraints,    event-triggered control, linear systems.
\end{IEEEkeywords}

\maketitle

\section{INTRODUCTION}

A networked control system with a feedback loop over a communication channel provides a first-order approximation of a cyber-physical system (CPS), where the interplay between the communication and control aspects of the system leads to new and unexpected analysis and design challenges~\cite{kumar,hespanha2007survey}\footnote{The material in this paper was presented in part at
the 18th European Control Conference (ECC), 2019~\cite{khojasteh2019stabilizing112}, and at 57th IEEE Conference on  Decision and Control, 2018~\cite{khojasteh2018estimating111}.}. In this setting, data-rate theorems quantify the impact of the communication channel on the ability to stabilize the system. Roughly speaking, these theorems state that stabilization requires a communication rate in the feedback loop at least as large as the intrinsic \emph{entropy rate} of the system, expressed by the sum of the logarithms of its unstable eigenvalues~\cite{Yukselbook,matveev2009estimation,Massimo,Nair,martins2006feedback,Delchamps,wong1999systems,baillieul1999feedback}.  

We consider a specific communication channel in the loop --- a \emph{timing channel}. Here, information is communicated through the timestamps of the symbols transmitted over the channel; the time is carrying the message. This formulation is motivated by recent works in event-triggering control, showing that the timing of the  triggering events carries information that can be used for 
stabilization~\cite{khojasteh2020exploiting,Level,OurJournal1,linsenmayer2017delay,yildiz2019event,cdc19papermj,guo2019optimal}.  By encoding information in timing,  stabilization can be achieved by transmitting additional data at a rate arbitrarily close to zero. However, in these works, the timing information was not explicitly quantified, and the analysis was limited to specific event-triggering strategies. In this paper, our goal is to determine the value of a timestamp from an information-theoretic perspective, when this timestamp is used for control.  We are further motivated by the results on the impact of multiplicative noise in control~\cite{ranade2013non, ranade2018control}, since timing uncertainty can lead to multiplicative noise in systems and thus can serve as an information bottleneck.

To illustrate the proof of concept that timing carries information useful for control, we consider the simple case of stabilization of a scalar, undisturbed, continuous-time, unstable, linear system over a timing channel and rely on the information-theoretic notion of \emph{timing capacity} of the channel, namely the amount of information that can be encoded using time stamps~\cite{anantharam1996bits, bedekar1998information,arikan2002reliability,sundaresan2000robust,rose2016inscribed1,gohari2016information,riedl2011finite,wagner2005zero,6736600,prabhakar2003entropy,aptel2018feedback,giles2002information,MaximRag,liu2004timing,sellke2007capacity}.
In this setting, the sensor can communicate with the controller by choosing the timestamps at which  symbols  from a unitary alphabet  are transmitted. The controller receives each transmitted symbol after a random delay is added to the timestamp.   We show the following data-rate theorem. For the state  to
converge to zero in probability, the timing capacity of the channel should be, essentially, at least
as large as the entropy rate of the system. Conversely, in the case the random delays are exponentially distributed, we show that when the timing capacity is strictly greater than the entropy rate of the system,
we can drive the state to zero in probability by using a decoder that   refines its estimate of the transmitted message every time a new symbol is received~\cite{como2010anytime}. 
We also derive analogous  necessary and sufficient conditions for the problem of estimating the state of the system with an  error that tends to zero in probability.

\begin{table*}[ht]  
  \begin{center}
    \caption{Capacity notions used to derive data-rate theorems in the literature under different notions of stability, channel types, and system disturbances.}
     \tabcolsep=0.11cm
       \begin{tabularx}{12.47cm}{l l l l l }
       \label{tab:table1}
       \\
        \hline 
      Work &  Disturbance &Channel  & Stability condition & Capacity   \\
      \hline \hline
      \cite{Mitter} & NO & Bit-pipe &   $|\XC(t)| \rightarrow 0$ a.s. & Shannon \\

      \cite{tatikonda2004control, matveev2007analoguedd} & NO & DMC &  $|\XC(t)| \rightarrow 0$ a.s. & Shannon     \\

      \cite{matveev2007shannon} & bounded & DMC & $ \mathbb{P}(\sup_{t} |\XC(t)|< \infty)=1$   & Zero-Error   \\
       \cite[Ch. 8]{matveev2009estimation}  & bounded & DMC & $    \mathbb{P}(\sup_t |\XC(t)|<K_\epsilon)>1-\epsilon$ & Shannon  \\
          \cite{sahai2006necessity}  & bounded & DMC & $   \sup_{t} \, \mathbb{E}(  |\XC(t)|^{m})<\infty$ & Anytime    \\
     \cite{nair2004stabilizability}  & unbounded & Bit-Pipe & $ \sup_{t }  \, \mathbb{E}( |\XC(t)|^{2})<\infty$& Shannon   \\
      \cite{Paolo,Lorenzo,minero2017anytime}  & unbounded & Var.  Bit-pipe & $ \sup_{t} \, \mathbb{E}(  |\XC(t)|^m)<\infty$ & Anytime   \\
         This paper  & NO &     Timing  &  ${|\XC(t)|\overset{P}{\rightarrow}0}$   & Timing   \\
         \hline   
    \end{tabularx}
  \end{center} 
%  \end{adjustwidth} 
\end{table*}

The books~\cite{Yukselbook,matveev2009estimation,fang2017towards,kawan2013invariance} 
and the surveys~\cite{Massimo,Nair,colonius2014analysis} provide detailed discussions of  data-rate theorems and related results that heavily inspire this work.
A portion of the literature studied stabilization over ``bit-pipe channels,'' where a rate-limited, possibly time-varying and erasure-prone communication channel is present in the feedback loop~\cite{Mitter,nair2004stabilizability,hespanha2002towards,Paolo,Lorenzo}. For  more general noisy channels, Tatikonda and Mitter \cite{tatikonda2004control} and Matveev and Savkin~\cite{matveev2007analoguedd} showed that  the state of undisturbed linear systems   can be forced to converge to zero almost surely (a.s.) if and only if   the Shannon capacity of the channel is larger than the entropy rate of the system.  
In the presence of disturbances, in order to keep the state bounded a.s., a more stringent condition is required, namely the zero-error capacity of the channel must be larger than the entropy rate of the system~\cite{matveev2007shannon}. Nair derived a similar information-theoretic result in a non-stochastic setting~\cite{girish2013}. Sahai and Mitter~\cite{sahai2006necessity} considered   moment-stabilization over noisy channels and in the presence of system disturbances of bounded support, and provided a  data-rate theorem in terms of the anytime capacity of the channel. They showed that to keep the $m$th moment of the state bounded, the anytime capacity of order $m$ should be larger than the entropy rate of the system. The anytime capacity has been further investigated in~\cite{ostrovsky2009error,sukhavasi2016linear,khina2016almost,minero2017anytime}.  Matveev and Savkin~\cite[Chapter~8]{matveev2009estimation} have also introduced a weaker  notion of stability in probability, requiring the state to be bounded with probability $(1-\epsilon)$ by a constant that diverges as $\epsilon \rightarrow 0$, and showed that in this case 
it is possible to stabilize linear systems with bounded disturbances over noisy channels provided that the Shannon capacity of the channel is larger than the entropy rate of the system. The various results, along with our contribution, are summarized in Table~\ref{tab:table1}. 
The main point that can be drawn from   all of these results  is that the relevant capacity notion for stabilization over a communication channel critically depends on the notion of stability and on the system's model.

From the system's perspective, our setup is closest to the one in~\cite{Mitter,tatikonda2004control, matveev2007analoguedd}, as there are no disturbances and the objective is to  drive the state to zero. Our convergence in probability provides a stronger necessary condition for stabilization, but a weaker sufficient condition than the one in these works.  We also point out that our notion of stability is  considerably stronger than the notion of probabilistic stability proposed in~\cite[Chapter~8]{matveev2009estimation}. 
Some additional works considered nonlinear  plants without disturbances~\cite{Topological,liberzon2005stabilization,de2005n},  and switched linear systems~\cite{liberzon2014finite,yang2018feedback}  where communication between the sensor and the controller occurs over a bit-pipe communication channel. The recent work in~\cite{sanjaroon2018estimation} studies   estimation of   nonlinear systems over noisy communication channels and the work in~\cite{kostina2016rate} investigates the trade-offs between the communication channel rate and the  cost of the linear quadratic regulator for linear plants.

Parallel work in control theory has investigated the possibility of stabilizing linear systems using timing information. 
One primary focus  of the emerging paradigm  of event-triggered control~\cite{Tabuada,WPMHH-KHJ-PT:12,astrom2002comparison,wang2011event,dimarogonas2012distributed,khashooei2018consistent,heemels2013periodic,li2012stabilizing,demirel2017trade,quevedo2014stochastic,lindemann2018event,girard2015dynamic,seuret2016lq}
has been on minimizing the number of transmissions while simultaneously ensuring the control objective~\cite{PT-JC:16-tac,pearson2017control}. 
Rather than performing periodic communication between  the system and the controller, in event-triggered control communication occurs only as needed, in an opportunistic  manner. In this setting, the timing of the triggering events can carry useful information about the state of the system, that can be used for stabilization~\cite{khojasteh2020exploiting,Level,OurJournal1,linsenmayer2017delay,yildiz2019event,cdc19papermj,guo2019optimal}. 
In this context, it has  been shown that the amount of timing information  is sensitive to the delay in the communication channel. While for small delay stabilization can be achieved  using only timing information and
transmitting data payload (i.e. physical data) at a rate
arbitrarily close to zero, for large values of the delay this is not the case, and the 
data payload rate
must be increased~\cite{OurJournal1,cdc19papermj}. In this paper, we  extend these results from an information-theoretic perspective, as we explicitly quantify the value of the timing information, independent of any transmission strategy. To quantify the amount of timing information alone, we  restrict to transmitting symbols from a unitary alphabet, i.e. at zero data payload rate. 
Research directions left open for future investigation include  the study of ``mixed'' strategies, using both timing information and physical data transmitted over a larger alphabet, as well as   generalizations to vector systems and the study of systems with disturbances. In the latter case, it is likely that the usage of stronger notions of capacity, or weaker notions of stability, will be necessary. 

The rest of the paper is organized as follows.
Section~\ref{sec:setup} introduces the system and channels models. The main results are presented in Section~\ref{sec:results}. Section~\ref{newsection111} considers the estimation problem, and Section~\ref{newsection222234fej} considers the stabilization problem.
Section~\ref{sec:relatedwork} provides a comparison with related work, and Section~\ref{simulations11} presents a numerical example. Conclusions are drawn in Section~\ref{sec:conc}.

%%%%%%%%%%%
%% Notation
%%%%%%%%%%%
\subsection{Notation}\label{sec:notation}
Let $X^n=(X_1,\cdots,X_n)$ denote a vector of random variables and let $x^n=(x_1,\cdots,x_n)$ denote its realization.
If  $X_{1}, \cdots, X_{n}$ are independent and identically distributed (i.i.d) random variables, then we refer to a generic $X_i \in X^n$ by $X$ and skip the subscript $i$. We use $\log$ and $\ln$ to denote the logarithms   base $2$ and base $e$ respectively. We use $H(X)$ to denote the Shannon entropy of a discrete random variable $X$ and $h(X)$ to denote the differential entropy of a continuous random variable $X$. Further, we use $I(X;Y)$ to indicate  the mutual information between random variables $X$ and $Y$.  
We write $X_n \xrightarrow{P} X$ if  $X_n$ converges in probability to $X$.  Similarly, we   write 
$X_n \xrightarrow{a.s.}X$ if $X_n$ converges almost surely to $X$.  For any set $\mathscr{X}$ and any $n \in \mathbb{N}$ we let
\begin{align}
\label{tunc-oper22}
    \pi_n: \mathscr{X}^\mathbb{N} \rightarrow \mathscr{X}^n
\end{align}
%$$ 
be the  truncation operator, namely the projection of a sequence in $\mathscr{X}^\mathbb{N}$ onto
its first $n$ symbols.

%%%%%%%%%%%
%% Problem formulation
%%%%%%%%%%%
\section{System and channel model}\label{sec:setup}
We consider the networked control system depicted in Fig.~\ref{fig:system}.  
\begin{figure}[t]
	\centering
 \includegraphics[width=.48 \textwidth]{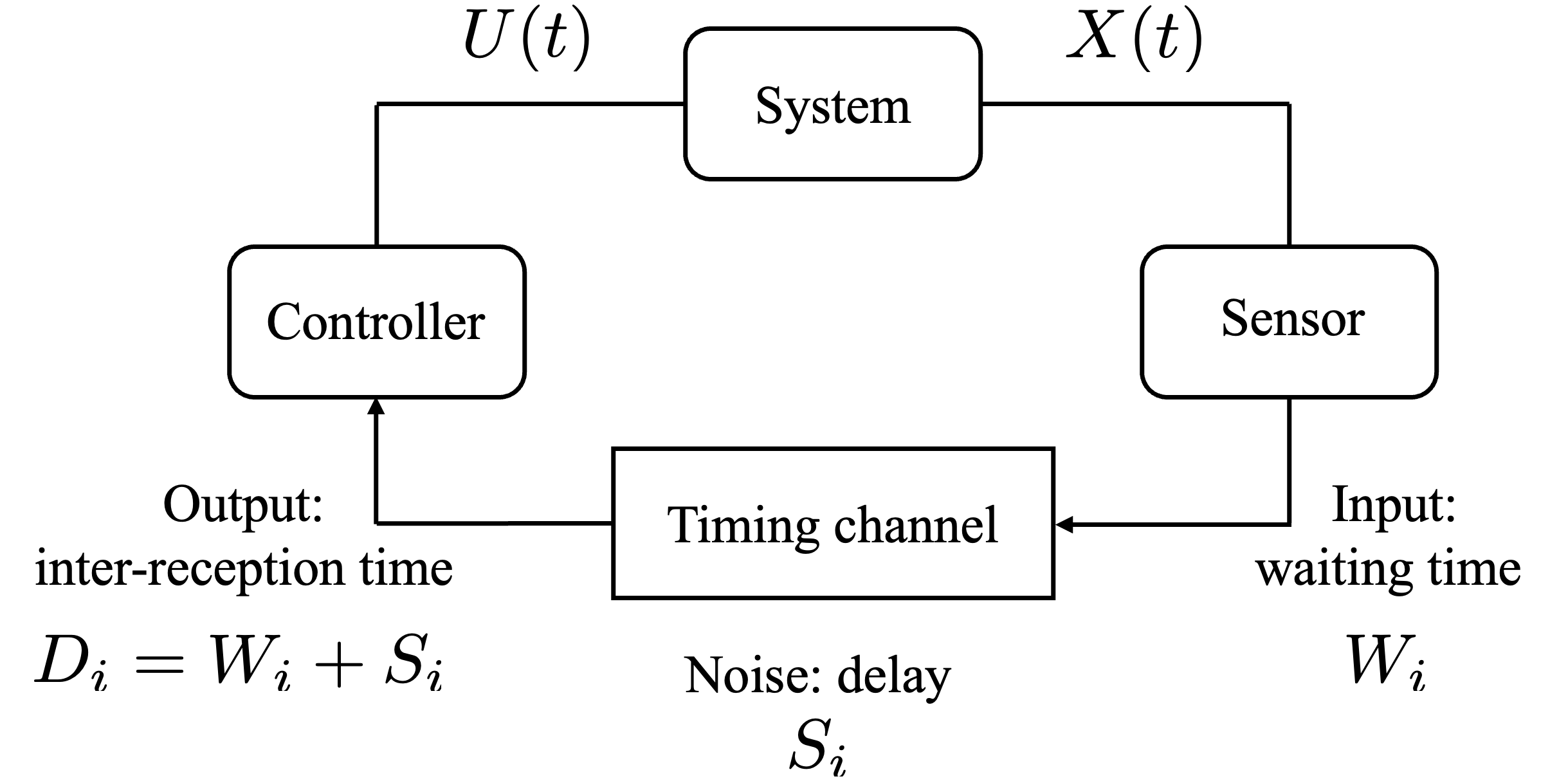}
	\caption{\footnotesize{Model of a networked control system where the feedback loop is closed over a timing channel~\eqref{channelm}.}}\label{fig:system}
\end{figure}
The system dynamics are described by a scalar, continuous-time, noiseless, linear time-invariant (LTI) system 
\begin{align}\label{syscon}
   	\dot{\XC}(t)=a \XC(t)+b \hspace{0.3mm} U(t),
\end{align}
where $\XC(t) \in \real$ and $U(t) \in \real$ 
are the system state and the control input respectively. The constants $a,b \in \mathbb{R}$ are such that $a>0$ and $b \not = 0$. The initial state  $X(0)$  is random and is drawn from a distribution of bounded differential entropy and bounded support, namely  ${h(X(0))<\infty}$   and $|X(0)|<L$, where $L$ is known to both the sensor and the controller.   
Conditioned on the realization of $X(0)$, the system evolution is deterministic. Both controller and sensor have knowledge of the system dynamics in \eqref{syscon}.       
We assume the sensor  can measure the state of the system with infinite precision, and the controller can apply the control input to the system with infinite precision and with zero delay. 

The sensor  is connected to the controller through a \emph{timing channel} (the telephone signaling channel defined in~\cite{anantharam1996bits}).
The operation of this channel is analogous to that of a  telephone system where a transmitter signals a phone call to the receiver through a ``ring'' and, after a random time required to establish the connection, is aware of the ``ring'' being received. Communication between transmitter and receiver can then occur without any vocal exchange, but by encoding messages in the ``waiting times'' between consecutive calls.

\subsection{The channel}
\label{sec:channeldefinition}
We model the channel as carrying symbols $\spadesuit$ from a  unitary alphabet, and each transmission is  received after a random delay. Every time a symbol is received, the sender is notified of the reception by an instantaneous acknowledgment. 
The channel is initialized with a $\spadesuit$ received at time $t=0$. After receiving  the acknowledgment for the $i$th $\spadesuit$, the sender waits for $W_{i+1}$ seconds and then transmits the next $\spadesuit$. Transmitted symbols are subject to i.i.d.\ random delays $\{S_i\}$.
Letting $D_i$ be the inter-reception time between two consecutive symbols, we have
\begin{align}\label{channelm}
D_i=W_i+S_i.
\end{align}
It follows that the reception time of the $n$th symbol is  
\begin{equation} \label{eq:taun}
\mathcal{T}_n=\sum_{i=1}^n D_i.
\end{equation}
Fig.~\ref{timing} provides an example of the timing channel in action.

\begin{figure*}[t]
	\centering
 \includegraphics[height=3.8cm,width=15.8cm]
 {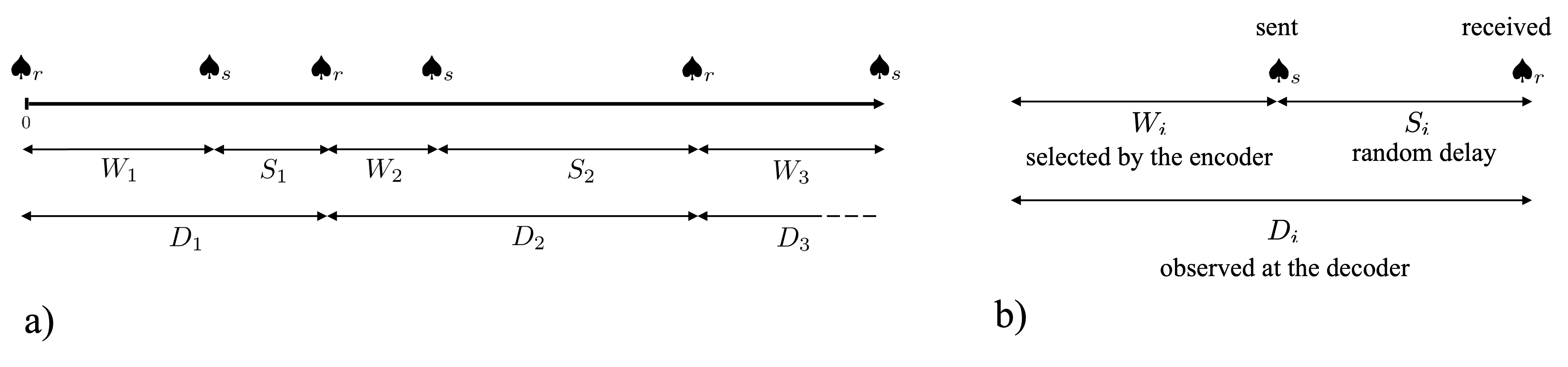}
 \centering
	\caption{\footnotesize{ a) The timing channel. b) $i$th transmission. Subscripts $s$ and $r$, for symbol $\spadesuit$, are used to denote sent and received symbols, respectively.  After reception of the $(i-1)$th symbol, the encoder selects a waiting time $W_i$, which carries information. The channel causes a random delay $S_i$, and the decoder observes the inter-reception time $D_i=W_i+S_i$\,.}}\label{timing}
\end{figure*}

\subsection{Capacity of the channel}
{We start by reviewing the following  definitions (Definitions 1-3) from~\cite{anantharam1996bits}.}
\begin{defn}
\label{randomcodebook}
 A  ${(n,M,T,\delta)}$-timing code for the telephone signaling channel consists of a codebook of $M$ codewords 
 $\{(w_{i,m},i=1,\ldots,n)$, $m=1 \ldots M\}$,
as well as a decoder, which upon observation of ${(D_1,\ldots,D_n)}$ selects the correct transmitted codeword with probability at least $1-\delta$. Moreover, the codebook is such that the expected random arrival time of  the $n$th symbol is at most $T$, namely
 \begin{align}\label{timingcapaman1}
 \expvalue\left(\mathcal{T}_n\right) \leq T.
 \end{align}
\end{defn}

\begin{defn} The rate of an $(n,M,T,\delta)$-timing code is
\begin{align}
 R = (\log M)/T.
 \end{align}
 \end{defn}

\begin{defn}
The timing capacity $C$ of the telephone signaling channel  is the
supremum of the achievable rates, namely the largest
$R$
such that for every $\gamma>0$ there exists a sequence of $(n,M_n,T_n,\delta_{T_n})$-timing codes
that satisfy
\begin{align}\label{timingcapaman2}
\frac{\log M_n}{T_n} > R-\gamma,
\end{align}
and $\delta_{T_n}\rightarrow 0$ as $n \rightarrow \infty$.
\end{defn}

The following result~\cite[Theorem 8]{anantharam1996bits} characterizes the capacity of the telephone signaling channel. 

\begin{theorem}[Anantharam and Verd\'{u}] The timing capacity of the telephone signaling channel is given by
\begin{align}
\label{capanaverdum}
C = \sup_{\chi>0} \sup_{\substack{W\ge 0 \\ \mathbb{E}(W)\le\chi}} \frac{I(W;W+S)}{\mathbb{E}(S)+\chi},
\end{align}
and if $S$ is exponentially distributed then
\begin{align}
\label{capanaverdumexp1}
C=\frac{1}{e \mathbb{E}(S)} \;\; \mbox{  \rm{[nats/sec]}.}
\end{align}
\end{theorem}

In this paper, 
we assume that
the waiting times $\{W_i\}$ used to encode any given message are generated at random in an i.i.d. fashion, and are also independent of the random delays $\{S_i\}$. Assuming the symbols in each codeword are picked i.i.d. from a common distribution restricts the encoder to using a fixed random telephoning policy. This assumption comes at no loss of generality since: (i) the capacity in \eqref{capanaverdum} is achieved by i.i.d. random codes ~\cite{anantharam1996bits}, and (ii) in our system model there are no disturbances and  therefore the control problem reduces to the communication of a fixed real-valued variable representing the initial condition with exponential reliability over a digital channel, which can be performed optimally using a fixed random coding strategy \cite{como2010anytime}.

\subsection{The sensor}
The sensor in Fig.~\ref{fig:system} can act as a source and channel encoder. Based on its source knowledge, namely the knowledge of the initial condition $X(0)$, system dynamics \eqref{syscon}, and  $L$, it   
selects the  waiting times $\{W_i\}$ between the reception and the transmission  of  consecutive  $\spadesuit$ symbols.
 
As in~\cite{anantharam1996bits,sundaresan2000robust} we assume that the causal acknowledgments received by the sensor every time a $\spadesuit$ is delivered to the controller are not used to choose the waiting times, but     only to avoid queuing, ensuring that every symbol is sent after the previous one has been received. This is analogous to TCP-based networks, where packet deliveries are
acknowledged via a feedback link~\cite{acknowledgement_fdbk,you2010minimum,yuksel2012random,gupta2009data,khina2018tracking}. 
For networked control systems, this causal acknowledgment can be obtained without assuming an additional communication channel in the feedback loop.
The controller can   signal the acknowledgment to the sensor by applying a control input to the system that excites a specific frequency of the state each time a symbol has been received. This strategy is known in the literature as ``acknowledgment through the control input''~\cite{tatikonda2004control,sahai2006necessity,matveev2009estimation,khojasteh2020exploiting},  and is used to avoid assuming an additional reverse communication channel in the feedback loop. In this setting, the causal acknowledgment can be generated through the plant itself. More specifically, the applied input may be viewed as the sum of a nominal stabilizing control component and a bounded acknowledgment marker (for example, a short pre-agreed excitation) inserted when a symbol is received.

\begin{remark}
Unlike conventional packet-based communication models, in which packets primarily carry payload data and network delay is treated as a property of the communication medium, the framework considered here is a timing-only communication model: transmitted symbols carry no payload, and information is encoded in the waiting times (and observed through the inter-reception times). The acknowledgment signal is used only to avoid queuing and is not intended to emulate TCP's  flow-control mechanisms.
\end{remark}

\subsection{The controller}
The controller in Fig.~\ref{fig:system} can act as a source and channel decoder.
It uses the reception times of all the symbols received up to time $t$, along with the knowledge of $L$ and of the system dynamics \eqref{syscon} to decode the source message, compute the control input $U(t)$, and apply it to the system. The control input can be refined over time, as the estimate of the source can be decoded with increasing accuracy when more and more symbols are received.
The objective is to design an encoding and decoding strategy to stabilize the system by driving the state to zero in probability, i.e.\ we want $\left|\XC(t)\right| \xrightarrow{P} 0$ as $t \to \infty$.  

Although the computational complexity of different encoding-decoding schemes is a key practical issue, in this paper we are concerned with the existence of schemes satisfying our objective, rather than with their  practical implementation.

%%%%%%%%%%%%%%%%%%%%%
%% Results
%%%%%%%%%%%%%%%%%%%%
\section{Main results} \label{sec:results}
\subsection{Necessary condition}
To derive a necessary  condition  for the stabilization of the feedback loop system depicted in Fig.~\ref{fig:system}, 
we first consider the problem of estimating the state in open-loop   over the timing channel along a specific sequence of estimation times. 
We show that if the estimation error tends to zero in probability along this sequence, then for all $\nu>0$ the timing capacity must be at least as large as $(\ngmm)$ times the entropy rate of the system.
This result holds for any source and random channel coding strategies  adopted by the sensor,  and for any strategy adopted by the controller to generate the control input.
Our proof employs a  rate-distortion argument to compute a lower bound on  the minimum number of bits required to represent the state   up to any given accuracy, and this leads to a corresponding lower bound on the required timing capacity of the channel. We then show that the same bound on the timing capacity  holds for   stabilization, since in order to have $\left|\XC(t)\right| \xrightarrow{P} 0$ as $t \to \infty$ in closed-loop,   the  estimation error  in open-loop   must   tend to zero in probability as $t \rightarrow \infty$, and therefore, in particular, along the designed sequence of estimation times.

\subsection{Sufficient condition}
To derive a sufficient condition for stabilization, we  first consider the problem of estimating the state in open-loop   over the timing channel. We focus on a specific sequence of estimation times. 
We provide an explicit source-channel coding scheme which guarantees that if for all $\nu>0$ the timing capacity is larger than $(\gmm)$  times the  entropy rate of the system, 
then the estimation error tends to zero in probability.
We then show that this condition is also sufficient to construct a control scheme such that $\left|\XC(t)\right| \xrightarrow{P} 0$ as $t \rightarrow \infty$.
The main idea behind our strategy is based on the realization that
in the absence of disturbances  all that is needed to drive the state to zero is  communicating the initial condition $X(0)$ to the controller with accuracy that increases exponentially over time. Once this is achieved, the controller can estimate the state $X(t)$ with increasing accuracy over time, and continuously apply an input that drives the state to zero.  This idea has been exploited before in the literature~\cite{Mitter,tatikonda2004control}, and the problem is related to the anytime reliable transmission of a real-valued variable over a digital channel~\cite{como2010anytime}.
Here, we  cast this problem in the framework of the timing channel.  
A main difficulty  in our case is to ensure that 
we can drive the system's state  to zero in probability despite the unbounded random delays occurring in the  timing channel.

In the source coding process, we quantize the interval $[-L,L]$ uniformly using a tree-structured quantizer~\cite{gersho2012vector}. We then map the obtained source code into a channel code suitable for transmission over the timing channel, using the capacity-achieving random codebook of~\cite{anantharam1996bits}. 
Given $X(0)$, the encoder picks a codeword from an arbitrarily large  codebook and starts transmitting the real numbers of the codeword one by one, where each real number corresponds to a holding time, and proceeds in this way forever.  
Every time a sufficiently large number of symbols are received, we use   a maximum likelihood decoder to successively refine the controller's estimate of $X(0)$.  
Namely,
the controller re-estimates $X(0)$ based on the new inter-reception times and all previous inter-reception times, and uses it  to compute the new state estimate of $\XC(t)$ and  control input $U(t)$. We show that when the sensor quantizes  $X(0)$ at sufficiently high resolution, and when the timing capacity is larger than the entropy rate of the system,    the   controller can construct a sufficiently accurate estimate of  $\XC(t)$ and compute $U(t)$ such that  $\left|\XC(t)\right| \xrightarrow{P} 0$ as $t \to \infty$.

 For ease of exposition, Appendix~\ref{flowchart-1} provides a roadmap of the main proofs, summarizing the logical steps underlying the necessary and sufficient conditions.

\section{The estimation problem}
\label{newsection111}
We  start considering the estimation problem depicted in Fig.~\ref{The observation problem}.
 \begin{figure}[t]
	\centering
 \includegraphics
 [width=8.1cm,height=1.40cm]
 {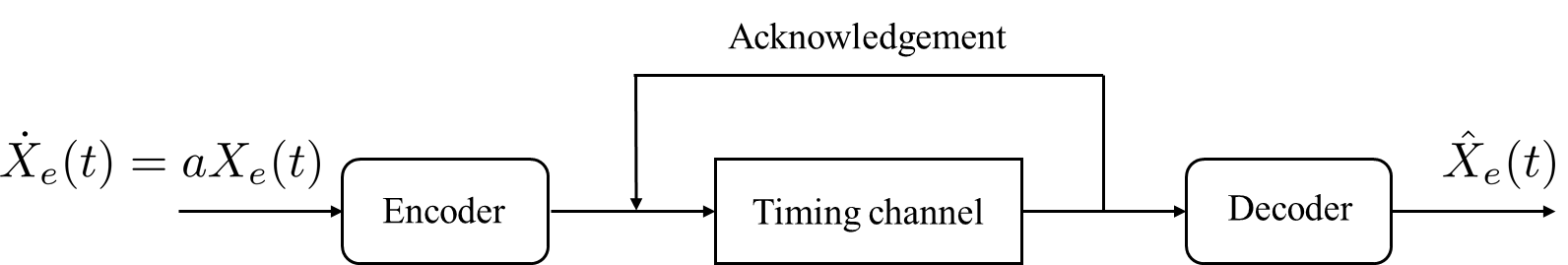}
	\caption{\footnotesize{The estimation problem.}}\label{The observation problem}
\end{figure}
By letting $b=0$ in~\eqref{syscon} we obtain the open-loop equation
\begin{align}
\dot{\XE}(t)=a\XE(t).
\label{estprob}
\end{align}
We assume that the encoder has  causal knowledge of the reception times via acknowledgments through the system as depicted in Fig.~\ref{The observation problem}. 
For any given $\nu>0$, consider a sequence of  estimation  times  $\timenN$  that satisfies 
\begin{align}
\label{tamoomnemishe3}
\ngmm  \le \lim_{n \rightarrow \infty} \frac{\timenN}{\expvalue(\mathcal{T}_n)} < 1.
\end{align}
Our first objective  is to obtain a necessary condition on the capacity of the timing channel required to  construct an estimate $\XEH(\timenN)$ such that 
$|\XE(\timenN)-\XEH(\timenN)|\overset{P}{\rightarrow} 0$ as $n \to \infty$.

Similarly, for any $\nu>0$ we also consider a sequence of  estimation  times 
$\timepS$ 
that satisfies 
\begin{align}
\label{tamoomnemishe2}
1 < \lim_{n \rightarrow \infty} \frac{\timepS}{\expvalue(\mathcal{T}_n)}\le \gmm .  
\end{align}
%we have
Our second objective  is to obtain a sufficient condition on the capacity of the timing channel that ensures the   construction of an  estimate $\XEH(\timepS)$  such that
$|\XE(\timepS)-\XEH(\timepS)|\overset{P}{\rightarrow} 0$ as $n \to \infty$.

The sequences of estimation times that satisfy~\eqref{tamoomnemishe3} and~\eqref{tamoomnemishe2} are ``close'' in the sense that
\begin{align}
\lim_{n \rightarrow \infty} \frac{\timepS-\timenN}{\expvalue(\mathcal{T}_n)} \le 2\nu.
\end{align}

Given the conditions in~\eqref{tamoomnemishe3} and~\eqref{tamoomnemishe2}, the next lemma provides  probabilistic bounds on the number of symbols that are received up to time $\timenN$ and $\timepS$, respectively. 
 Note that, in the limit, $\timenN$ is chosen slightly before the average reception time $\expvalue(\mathcal{T}_n)$, and $\timepS$ is chosen slightly after $\expvalue(\mathcal{T}_n)$. Here, $\nu$ measures the asymptotic timing slack around the average reception time. In essence, the next lemma proves that, as $n \rightarrow \infty$, by time 
$\timenN$, at most $n$ symbols have arrived, and by time $\timepS$, at least $n$ symbols have arrived with arbitrarily high probability.
\begin{lemma}
\label{lemnewsub20}
 Given the conditions in~\eqref{tamoomnemishe3}, the probability $\mathbb{P}(  \mathcal{T}_{n+1} \le \timenN)$ tends to zero as $n \rightarrow \infty$. Moreover, given the condition in~\eqref{tamoomnemishe2}, 
 the probability $\mathbb{P} ( \mathcal{T}_n > \timepS  )$ tends to zero as  $n \rightarrow \infty$. 
\end{lemma}
\begin{proof}
We start by proving that $\mathbb{P}(\mathcal{T}_{n+1} \le \timenN)$ tends to zero as $n \rightarrow \infty$. 
For large enough $n$, using~\eqref{tamoomnemishe3}, we have   that the probability of receiving the $n+1$ symbols before $\timenN$ is
\begin{align}
  \nonumber
\mathbb{P}(\mathcal{T}_{n+1} \le \timenN) &\leq \mathbb{P}(\mathcal{T}_{n+1}/(n+1) < \mathbb{E}(\mathcal{T}_n)/(n+1))
\\
&\le \mathbb{P}(\mathcal{T}_{n+1}/(n+1) < \mathbb{E}(D)).
  \label{tendtozero3}
\end{align}

Since the waiting times $\{W_i\}$ and the random delays $\{S_i\}$ are i.i.d.\ sequences and independent of each other, it follows   by the strong law of large numbers that \eqref{tendtozero3} tends to zero as $n \rightarrow \infty$.

We continue 
by bounding the probability of the event that the $n$-th symbol does not arrive by the estimation deadline $\timepS$. 
For large enough $n$, using~\eqref{tamoomnemishe2}, we have   that the probability of missing the deadline is
\begin{align}
  \nonumber
\mathbb{P}(\mathcal{T}_n > \timepS) &\leq \mathbb{P}(\mathcal{T}_n/n > \mathbb{E}(\mathcal{T}_n)/n)
\\
&=\mathbb{P}(\mathcal{T}_n/n > \mathbb{E}(D)).
  \label{tendtozero}
\end{align}
As in~\eqref{tendtozero3}, because the waiting times $\{W_i\}$ and delays $\{S_i\}$ are i.i.d.\ and independent, the strong law of large numbers ensures that \eqref{tendtozero} converges to zero as $n \to \infty$.
\end{proof}
Lemma~\ref{lemnewsub20} leads to the following conclusions. First, with high probability 
as $n \rightarrow \infty$, by time $\timenN$ at most $n$ symbols have been received.    Second, with high probability 
as $n \rightarrow \infty$ the estimation at time $\timepS$ is evaluated,  after  at least  $n$ symbols have been received.

 \subsection{Necessary condition} 

The next theorem provides a necessary condition on the timing capacity
for the state estimation error to tend to zero in probability. 
\begin{theorem}\label{NECCONDITIO}
Assume that $W_i, S_i$ are i.i.d.
Consider the estimation problem  depicted in Fig.~\ref{The observation problem} with system dynamics~\eqref{estprob}. Consider transmitting $n$ symbols over the telephone signaling channel~\eqref{channelm}, and for any $\nu>0$ let  the sequence of estimation times  satisfy  \eqref{tamoomnemishe3}. 
If ${|\XE(\timenN)-\XEH(\timenN)|\overset{P}{\rightarrow}0}$, then 
\begin{align}\label{mainresult}
I(W;W+S)  \ge ~a~(\ngmm)~\expvalue(W+S) \;\;\; [\mbox{\rm{nats}}],
\end{align}
and consequently
\begin{align}\label{obs:cap:nec:12}
C \ge  a \, (\ngmm)~~[\mbox{\rm{nats}}/\mbox{\rm{sec}}].
\end{align}
\end{theorem}

The proof of Theorem~\ref{NECCONDITIO} is given in the appendix.

\hfill \break
\subsection{Sufficient condition}
The next theorem provides a sufficient condition for convergence of the state estimation error to zero in probability along  any sequence of estimation times $\timepS$ satisfying~\eqref{tamoomnemishe2}, in the case of exponentially distributed delays.

\begin{theorem}\label{SUFFCONDMA}
For any given $\nu > 0$.
Consider the estimation problem  depicted in Fig.~\ref{The observation problem} with system dynamics~\eqref{estprob}. Consider transmitting $n$ symbols over the telephone signaling channel~\eqref{channelm}.
Assume $\{S_i\}$ are drawn i.i.d. from exponential distribution with mean $\expvalue(S)$. 
If the capacity of the timing channel  is at least
\begin{align}\label{sufCap2}
C  \ge a \, (\gmm) \;\;\;[\mbox{\rm{nats}}/\mbox{\rm{sec}}],
\end{align}
then for any sequence of times $\{\timepS \}$ that satisfies~\eqref{tamoomnemishe2}, 
we can compute an estimate $\XEH(\timepS)$ such that as $n \to \infty$, we have
\begin{align}
  {|\XE(\timepS)-\XEH(\timepS)|\overset{P}{\rightarrow}0}.
\end{align}
\end{theorem}

The proof  of  Theorem~\ref{SUFFCONDMA} is given in the appendix. 
The result is strengthened in the next section (see Corollary~\ref{remark2}), showing that   $C>a \, (1+\nu)$ is also sufficient to drive the state estimation error     to zero in probability for all $t \rightarrow \infty$.

\begin{remark}
 Since $\nu >0$ can be chosen arbitrarily small 
Theorems~\ref{NECCONDITIO} and~\ref{SUFFCONDMA} provide an ``almost" tight necessary and sufficient condition for the estimation problem.
The entropy-rate of our system is   $a$ nats/time~\cite{Topological,liberzon2017entropy,colonius2009invariance,colonius2012minimal,rungger2017invariance}. This represents the amount of uncertainty per unit time generated by the system in open loop. In fact, 
\eqref{mainresult},
can be seen as a typical scenario in    data-rate theorems: to drive the error to zero the mutual information between an  encoding symbol $W$ and its received  noisy version $W+S$ should be larger than the average ``information growth'' of the state during the inter-reception interval $D$, which is given by
\begin{align}
\expvalue(a D)= a~\expvalue(W+S).
\end{align} 
\end{remark}

\section{The stabilization problem}
\label{newsection222234fej}
\subsection{Necessary condition}
We now turn to consider the stabilization problem. Our first lemma states that if in closed-loop  we are able to drive the state to zero in probability, then in open-loop we are also able to estimate the state with vanishing error in probability. 
\begin{lemma}\label{ocem234221q}
Consider   stabilization of the closed-loop system~\eqref{syscon}   and  estimation of the open-loop system~\eqref{estprob} over the timing channel~\eqref{channelm}. If there exists a controller such that  ${|\XC(t)|\overset{P}{\rightarrow} 0}$ as $t \rightarrow \infty$, in closed-loop, then  there exists an estimator such that
$|\XE(t )-\XEH(t )|\overset{P}{\rightarrow} 0$ as $t \rightarrow \infty$, in open-loop.
\end{lemma}
\begin{proof}
From~\eqref{syscon}, we have in closed loop
    \begin{align} 
      \label{zeta}
    \XC(t) &= e^{at}X(0)+\zeta(t),  \\  
    \zeta(t)&=e^{at} \int_{0}^{t}  \label{zeta1}
    {e^{-a\varrho}bU(\varrho)d\varrho}. 
  \end{align}
It follows that if
\begin{align}
\lim_{t \rightarrow \infty} \mathbb{P}\left(|\XC(t)| \le \epsilon\right)=1,
\end{align}
then we also have  
\begin{align}\label{fkekw2224r}
\lim_{t \rightarrow \infty} \mathbb{P}\left(\left|e^{at}X(0)+\zeta(t)\right| \le \epsilon\right)=1.
\end{align}
On the other hand, from~\eqref{estprob}  we have in open loop
\begin{align}
\XE(t) &= e^{at}X(0),
\end{align}
and we can   choose $\XEH(t) = -\zeta(t)$  so that   
\begin{align}
|\XE(t)-\XEH(t)|=|e^{at}X(0)+\zeta(t)|\overset{P}{\rightarrow}0,
\end{align}
where the last step follows from~\eqref{fkekw2224r}. 
\end{proof}

The next theorem provides a necessary rate for the stabilization problem. 
\begin{theorem}\label{NECCONDITIO:stabili}
Consider the stabilization of the closed-loop system~\eqref{syscon}. 
If ${|\XC(t)|\overset{P}{\rightarrow}0}$ as $t \to \infty$, then
\begin{align}\label{nec:stab:3202}
I(W;W+S) ~ \ge ~ a~(\ngmm)~\expvalue(W+S) \;\;\;  [\mbox{\rm{nats}}],
\end{align}
and consequently 
\begin{align}\label{nec:stab:3202:cap}
C ~ \ge  a~(\ngmm) ~~[\mbox{\rm{nats}}/\mbox{\rm{sec}}].
\end{align}
\end{theorem}
\begin{proof}
By Lemma~\ref{ocem234221q} we have that if ${|\XC(t)|\overset{P}{\rightarrow}0}$, then ${|\XE(t)-\XEH(t)|\overset{P}{\rightarrow} 0}$ for all $t \rightarrow \infty$, and in particular along a sequence
$\{\timenN\}$ satisfying \eqref{tamoomnemishe3}.  The result now follows from Theorem~\ref{NECCONDITIO}.
\end{proof}

\subsection{Sufficient condition}
Our next lemma strengthens our estimation results, stating that it is enough for the state estimation error to converge to zero in probability as $n \rightarrow \infty$ along a sequence of estimation times $\{\timepS\}$ satisfying~\eqref{tamoomnemishe2}, to ensure it   converges to zero for all $t \rightarrow \infty$.
\begin{lemma}\label{vri221333}
Consider   estimation of the system~\eqref{estprob} over the timing channel~\eqref{channelm}.
If there exists  $1<\Gamma_0 \leq \gmm$ such that along the sequence of estimation times $\timepS = \Gamma_0 \mathbb{E}(\mathcal{T}_n)$ (which satisfies~\eqref{tamoomnemishe2})  we have  ${|\XE(\timepS)-\XEH(\timepS)|\overset{P}{\rightarrow} 0}$ as $n \rightarrow \infty$, then for all  $t \rightarrow \infty$ we also have ${|\XE(t)-\XEH(t)|\overset{P}{\rightarrow} 0}$.
\end{lemma}

\begin{proof}
We have that for   $\timepS=\Gamma_0 \mathbb{E}(\mathcal{T}_n)$ and for all   $\epsilon'>0$,   and  $\phi>0$, there exist $n_\phi$ such that for all $n\ge n_\phi$
\begin{align}\label{wwIEE2nkmkkmlo}
\mathbb{P}\left(|\XE(\timepS)-\XEH(\timepS)|>\epsilon'\right)\le \phi.
\end{align}
Let $t_{n_\phi}=\Gamma_0 \mathbb{E}(\mathcal{T}_{n_\phi})$ be the time at which we estimate the state for the $n_\phi$th time. We want to show that for all $t \in [t_{n_\phi},t_{n_{\phi}+1}]$ and $\epsilon>0$, we also have
\begin{align}\label{wo13rwx111}
\mathbb{P}\left(|\XE(t)-\XEH(t)|>\epsilon\right)\le \phi.
\end{align}
Consider the random time 
$\mathcal{T}_{n_\phi}$ at which   $\spadesuit$ is received for the $n_\phi$th   time.
We have
\begin{align} \label{timediff}
t_{n_{\phi}+1}-t_{n_\phi} &= \Gamma_0 \, \expvalue(\mathcal{T}_{n_{\phi}+1}) -  \Gamma_0 \, \expvalue(\mathcal{T}_{n_{\phi}})   \nonumber \\
&= (n_\phi+1) \Gamma_0 \, \expvalue(D) -  n_\phi \Gamma_0 \, \expvalue(D)  \nonumber \\
& =  \Gamma_0 \, \expvalue(D).
\end{align}

 For all  $t \in [t_{n_\phi},t_{n_{\phi}+1}]$, from the open-loop equation~\eqref{estprob} we have
\begin{align} \label{comb1}
\XE(t) &= e^{a(t-t_{n_\phi})}\XE(t_{n_\phi}).
\end{align}
We then let
\begin{align} \label{comb2}
\XEH(t) &= e^{a(t-t_{n_\phi})}\XEH(t_{n_\phi}).
\end{align}
Combining \eqref{comb1} and \eqref{comb2} and using~\eqref{timediff}, we obtain  that for all  $t \in [t_{n_\phi},t_{n_{\phi}+1}]$   
\begin{align}\label{nfee2ed222}
|\XE(t)-\XEH(t)| & \le e^{a \Gamma_0 \, \expvalue(D)}|\XE(t_{n_\phi})-\XEH(t_{n_\phi})|.
\end{align}
From which it follows that
\begin{align}\nonumber
\mathbb{P}\left(|\XE(t)-\XEH(t)| >\epsilon'e^{a \Gamma_0 \, \expvalue(D) }\right) 
\\
 \leq \mathbb{P}\left(|\XE(t_{n_\phi})-\XEH(t_{n_\phi})|> \epsilon'\right).
\end{align}
Since~\eqref{wwIEE2nkmkkmlo} holds for all $n \geq n_\phi$, we also  have
\begin{align}
\mathbb{P}\left(|\XE(t_{n_\phi})-\XEH(t_{n_\phi})|\geq \epsilon'\right) \leq \phi.
\end{align}

We can now let $\epsilon'< \epsilon e^{-a \Gamma_0 \, \expvalue(D)}$    and the result follows.
\end{proof}
\hfill \break
 
Lemma~\ref{vri221333} yields the following corollary, which is an immediate extension of Theorem~\ref{SUFFCONDMA}.
\begin{corollary}~\label{remark2}
Consider the estimation problem  depicted in Fig.~\ref{The observation problem} with system dynamics~\eqref{estprob}. Consider transmitting $n$ symbols over the telephone signaling channel~\eqref{channelm}.
Assume $\{S_i\}$ are drawn i.i.d. from exponential distribution with mean $\expvalue(S)$. 
If the capacity of the timing channel  is at least $C \ge a \, (\gmm)$, then we have ${|\XE(t)-\XEH(t)|\overset{P}{\rightarrow}0}$ as $t \rightarrow \infty$. 
\end{corollary}
\begin{proof}
We start by considering the  sequence of estimation times $\timepS= (\gmm)\, \mathbb{E}(\mathcal{T}_n)$. Since 
$C \ge a \, (\gmm)$, by Theorem~\ref{SUFFCONDMA}   
we have ${|\XE(\timepS)-\XEH(\timepS)|\overset{P}{\rightarrow} 0}$ as $n \rightarrow \infty$. Then, by Lemma~\ref{vri221333}   we also have ${|\XE(t)-\XEH(t)|\overset{P}{\rightarrow} 0}$ as $t \rightarrow \infty$.
\end{proof}

The next key lemma states that if   we are able to estimate the state with vanishing error in probability, then we are also able to drive the state to zero in probability. 
\begin{lemma}\label{iirir3404www}
Consider   stabilization of the closed-loop system~\eqref{syscon}   and  estimation of the open-loop system~\eqref{estprob} over the timing channel~\eqref{channelm}.
If there exists an estimator such that ${|\XE(t )-\XEH(t )|\overset{P}{\rightarrow} 0}$  as $t \rightarrow \infty$, in open-loop,  then  there exists a controller such that  ${|\XC(t)|\overset{P}{\rightarrow} 0}$  as $t \to \infty$, in closed-loop.
\end{lemma}
\begin{proof}
We start by showing that if there exists an open-loop estimator    such that $|\XE(t)-\XEH(t )|\overset{P}{\rightarrow} 0$  as $t \rightarrow \infty$,  then  there also  exists a   closed-loop estimator  such that  ${|\XC(t)-\XCH(t)|\overset{P}{\rightarrow} 0}$  as $t \to \infty$. We construct the closed-loop estimator based on the open-loop estimator as follows.
The sensor in closed-loop runs a copy of the open-loop system by constructing the virtual open-loop dynamic
\begin{align}
\label{kewoef!!!33332322}
    \XE(t)=X(0)e^{at}.
\end{align}
Using the open-loop estimator, for all $t>0$ the controller  acquires the open-loop estimate $\XEH(t)$ such that ${|\XE(t)-\XEH(t )|\overset{P}{\rightarrow} 0}$. It then uses this estimate to construct the closed-loop estimate
\begin{align}\label{!!2223i3n33!} 
    \XCH(t) &= 
    \XEH(t)+e^{at} \int_{0}^{t} 
    {e^{-a \varrho}bU( \varrho)d \varrho}.
  \end{align}
Since from~\eqref{syscon} the true state in closed loop is
\begin{align} 
      \label{zetax}
    \XC(t) &= X(0)e^{at} +e^{at} \int_{0}^{t} {e^{-a\varrho}bU(\varrho)d\varrho},
  \end{align}
  it follows by combining~\eqref{kewoef!!!33332322}, \eqref{!!2223i3n33!} and \eqref{zetax} that
  \begin{align}
  \label{efjijie!!!2234441}
|\XC(t)-\XCH(t)|=|\XE(t)-\XEH(t)|\overset{P}{\rightarrow} 0.
\end{align}

What remains to be proven is that 
if ${|\XC(t)-\XCH(t)|\overset{P}{\rightarrow} 0}$, then there exists a controller such that ${|\XC(t)|\overset{P}{\rightarrow} 0}$.

Let $b>0$ and  
choose  $k$  so large that $a-bk<0$. Let $U(t)=-k \XCH(t)$. From~\eqref{syscon}, we have
\begin{align}\label{slove321}
  \dot{\XC}(t)=(a-bk)\XC(t)+bk[\XC(t)-\XCH(t)].
\end{align}
By solving~\eqref{slove321} and using the triangle inequality, we get
 \begin{align}\label{mmd2344i}
  |\XC(t)|\le &|e^{(a-bk)t}X(0)|+
  \nonumber \\
  &\left|\int_{0}^{t}e^{(t-\varrho)(a-bk)}bk(\XC(\varrho)-\XCH(\varrho))d\varrho\right|.
  \end{align}
Since ${|X(0)|<L}$ and ${a-bk<0}$,  the first term in~\eqref{mmd2344i} tends to zero as $t \rightarrow \infty$. Namely, for any ${\epsilon>0}$  
there exists a number $N_{\epsilon}$ such that for all $t \ge N_{\epsilon}$, we have 
\begin{align}
{|e^{(a-bk)t}X(0)|\le \epsilon}.
\end{align}
Since by \eqref{efjijie!!!2234441} we have that ${|\XC(t)-\XCH(t)| \overset{P}{\to} 0}$, we also have that for any $\epsilon, \delta>0$ 
there exist a   number $N'_{\epsilon}$ such that for all ${t \ge N'_{\epsilon}}$,  we have
\begin{align}
{\mathbb{P}\left(|\XC(t)-\XCH(t)|\le \epsilon\right) \geq 1-\delta}.
\end{align}
It now follows from \eqref{mmd2344i} that for all ${t \ge \max\{N_{\epsilon},N'_{\epsilon}\}}$   the following inequality holds with probability at least $(1-\delta)$
 \begin{align}\label{mm32vgds}
  \left|\XC(t)\right|&\le \epsilon +bke^{t(a-bk)}\int_{0}^{{N'_\epsilon}}e^{-\varrho(a-bk)} |\XC(\varrho)-\XCH(\varrho) |d\varrho~
  \nonumber \\
 &+\epsilon bke^{t(a-bk)}\int_{{N'_{\epsilon}}}^{t}e^{-\varrho(a-bk)}d\varrho.
  \end{align} 
Since both sensor and controller are aware that $|X(0)| < L$,  by \eqref{kewoef!!!33332322} we have that for all $t\geq 0$ the open-loop estimate  acquired by the controller satisfies $\XEH(t) \in [-Le^{at},Le^{at}]$. By~\eqref{efjijie!!!2234441} the closed-loop estimation error is the same as the open-loop estimation error, and we then have
that for all ${\varrho \in [0,{N'_{\epsilon}}]}$ 
\begin{align} \label{subs}
{|\XC(\varrho)-\XCH(\varrho)| = |\XE(\varrho)  -  \XEH(\varrho)| \le 2L e^{a {N'_{\epsilon}}}}.
\end{align}
%a.s
Substituting \eqref{subs} into~\eqref{mm32vgds}, we obtain that with probability at least $(1-\delta)$
\begin{align}
  \left|\XC(t)\right|\le &\epsilon+ 2Lbke^{[t(a-bk)+ aN'_\epsilon]}\frac{e^{-{N'_{\epsilon}}(a-bk)}-1}{-(a-bk)} 
  \nonumber \\
 &+ \epsilon bke^{t(a-bk)}~\frac{e^{-t(a-bk)}-e^{-{N'_{\epsilon}}(a-bk)}}{-(a-bk)}. \label{letting1} %~~\mbox{w.h.p.}
  \end{align}
By first letting $\epsilon$ be sufficiently close to zero, and then letting $t$ be sufficiently large, we can make the right-hand side of \eqref{letting1} arbitrarily small, and the result follows. 
\end{proof}

The next theorem combines the results above, providing a sufficient condition for convergence of the state to zero in probability in the case of exponentially distributed delays.
\begin{theorem}\label{SUFFCONDMA:stability}
Consider the stabilization of the system~\eqref{syscon}. 
Assume $\{S_i\}$ are drawn i.i.d. from an exponential distribution with mean $\expvalue(S)$. 
If the capacity of the timing channel is at least
\begin{align}\label{suf:stab3002}
C \ge  a \, (\gmm)\;\;\;[\mbox{\rm{nats}}/\mbox{\rm{sec}}],
\end{align}
then ${|\XC(t)|\overset{P}{\rightarrow}0}$ as $t \to \infty$. 
\end{theorem}

\section{Comparison with previous work} \label{sec:relatedwork}
\subsection{Comparison with stabilization over an erasure channel}\label{sec:erasure}
In~\cite{tatikonda2004control}   the problem of stabilization of the discrete-time version of the system in~\eqref{syscon}   
over an erasure channel has been considered. In this discrete model, at each time step of the system's evolution the sensor transmits    $\bar{I}$ bits 
to the controller and these bits are successfully  delivered   with probability $1-\mu$, or they are dropped with probability $\mu$, in an independent fashion. It is shown that a necessary condition for ${X(k) \xrightarrow{a.s} 0}$   is  that the capacity of this $\bar{I}$-bit erasure channel is
\begin{align}\label{erestat}
(1-\mu) \bar{I} \ge \log a~~[\mbox{bits}/\mbox{sec}] {\,=\, \ln 2 \, \log a ~~[\mbox{nats}/\mbox{sec}]}.
\end{align}
Since almost sure convergence implies convergence in probability, by  Theorem~\ref{NECCONDITIO:stabili} we have that  
 the following necessary condition holds in our setting  for ${X(t) \xrightarrow{a.s.} 0}$:  
\begin{align}\label{similaregu}
\frac{I(W;W+S)}{~\expvalue(W+S)}  \ge  a \, (\ngmm) ~~~~[\mbox{nats/sec}],
\end{align}
where $\nu>0$ can be arbitrarily small.

We now compare \eqref{erestat} and~\eqref{similaregu}.
The rate of expansion of the state space of the continuous system in open loop is $a$ nats per unit time, while for the discrete system is $\log  a$ bits per unit time. Accordingly, \eqref{erestat} and~\eqref{similaregu} are parallel to each other: in the case of~\eqref{similaregu}  the controller   must receive at least $a \, \expvalue(W+S)$ nats   representing the initial state  during a time interval of average length $\expvalue(W+S)$. In the case of~\eqref{erestat}  the controller must receive at least $\log a/(1-\mu)$ bits representing the initial state over  a  time interval whose average length corresponds to the average number of trials before the first successful reception 
\begin{align}
(1-\mu)\sum_{k=0}^{\infty} (k+1) \mu^k =\frac{1}{1-\mu}.
\end{align}

\subsection{Comparison with event triggering strategies}
The works~\cite{khojasteh2020exploiting,Level,OurJournal1,linsenmayer2017delay,yildiz2019event,cdc19papermj,guo2019optimal}
use event-triggering strategies that exploit timing information for stabilization over a digital communication channel. 
These    strategies encode information over time in a specific state-dependent fashion and use a combination of timing information and data payload to convey information used for stabilization.  

Our framework, by considering the transmission of symbols from a unitary alphabet, uses only timing information for stabilization. In
Theorem~\ref{NECCONDITIO:stabili} we provide a fundamental limit on the rate at which information can be encoded in time, independent of any transmission strategy.  
Theorem~\ref{SUFFCONDMA:stability} then shows that this limit can be almost achieved, in the case of  exponentially distributed delays.

The work~\cite{Level} shows that using event triggering it is possible to achieve stabilization    with any positive transmission rate over a zero-delay digital communication channel. 
Indeed, for channels without delay achieving stabilization at zero rate is easy. One could for example
transmit a single  symbol at a time  equal to
any bijective mapping of $x(0)$ into a point of the non-negative reals. For example, we could  transmit $\spadesuit$  at time $t=\tan^{-1}(x(0))$ for $t \in [0,\pi]$.
The reception of the symbol would reveal the initial state exactly, and the system could be stabilized.

The work in~\cite{OurJournal1} shows that when the delay is positive, but sufficiently small, a triggering policy can still achieve stabilization with any positive transmission rate. However, as the delay increases  past a critical threshold, the timing information becomes so much out-of-date that the transmission rate must begin to increase. 
In our case, since the capacity of our timing channel depends on the distribution of the delay, we may also expect that 
a large value of the capacity, corresponding to a small average delay, would allow for stabilization to occur using only timing information. Indeed, when delays are distributed exponentially, 
from~\eqref{capanaverdumexp1} and 
Theorem~\ref{SUFFCONDMA:stability} it follows that  as longs as the expected value of delay is
\begin{align}
\label{eq:cirticaldelay-timing}
\expvalue(S) < \frac{1}{e \, a},
\end{align}
it is possible to stabilize the system by using only   timing information. On the other hand, the system is not stabilizable using only   timing information if the expected value of the delay becomes larger than $(e \, a)^{-1}$.

\begin{remark}
In this paper we consider a timing channel in which transmitted symbols carry no data payload, and all information is conveyed through the inter-reception times. This contrasts with conventional digital communication models (e.g., discrete memoryless channels), where information is conveyed through payload bits and transmission times are typically not used for encoding. Event-triggered control mechanisms such as~\cite{OurJournal1} lie between these two cases: transmissions occur at state-dependent times, and information is conveyed jointly through the payload and the choice of transmission times.

Additionally, the critical-delay phenomenon reported in~\cite{OurJournal1} (see also our survey~\cite{franceschetti2023information}) is derived under a different delay model, in which the communication delay is unknown but bounded by an upper bound $\gamma'$. In that setting, the required payload transmission rate exhibits a phase transition: for $\gamma' < \gamma_c$, stabilization is possible with arbitrarily small payload rate by exploiting timing information carried by the event times, whereas for $\gamma' > \gamma_c$ a strictly positive payload rate becomes necessary 
because delay uncertainty makes the timing information insufficient, by itself, to support stabilization.
For scalar plants, the critical bound satisfies $a\,\gamma_c=\ln 2$, i.e., timing alone suffices as long as the worst-case delay bound times the plant's entropy rate remains below $\ln 2$. In contrast, our timing-channel condition in~\eqref{eq:cirticaldelay-timing} can be written as $a\,\mathbb{E}[S]< 1/e$ for exponential delays. Thus, both thresholds depend inversely on the plant's entropy rate, but the constants differ because they correspond to different notions of delay (worst-case bound $\gamma'$ versus mean delay $\mathbb{E}[S]$) and different communication abstractions (event-triggered strategies versus strategy-independent timing capacity).
\end{remark}

\section{Numerical example}
\label{simulations11}
\begin{figure*}[t]
	\centering
 \includegraphics[scale=0.55]
 {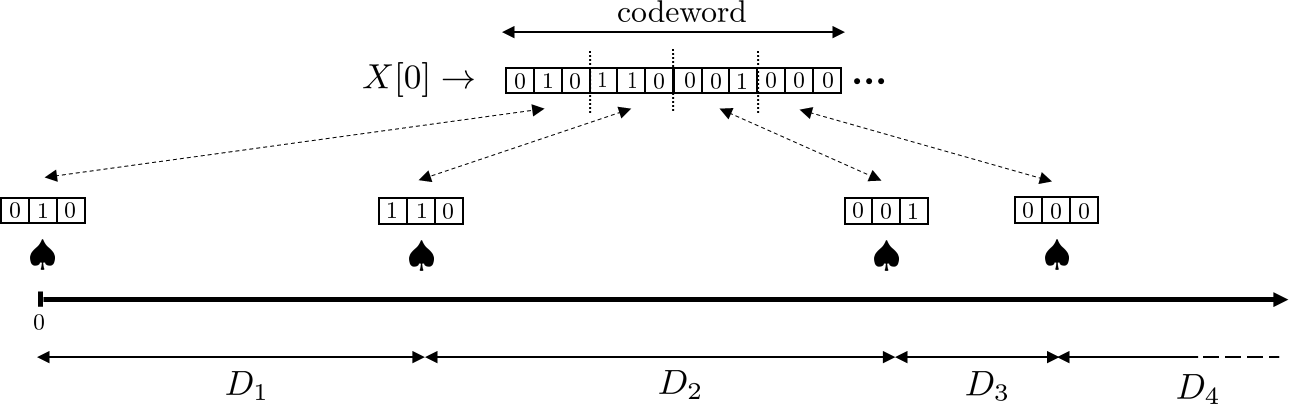}
	\caption{\footnotesize{Evolution of the channel used in the simulation in an error-free case.} Each time $\spadesuit$ is received,  a new codeword is decoded  using all the symbols received up to that time. The decoded codeword represents the initial state $X[0]$ with a precision that increases by $\mathbb{E}(D)C$ bits at each symbol reception. In the figure, for illustration purposes we have assumed $\mathbb{E}(D)C=3$ bits. }\label{timing_simulations3}
\end{figure*}
\label{simsec}
We now present a numerical simulation of stabilization over the telephone signaling channel. While our analysis is for continuous-time
systems, the simulation is  performed in discrete time, considering the system
\begin{align}
\XC[m+1]=a\XC[m]+U[m],~~\mbox{for} ~~m\in \integerspos,
\end{align}
where $a>1$ so that the system is unstable.

In this case, assuming i.i.d.\ geometrically distributed delays $\{S_i\}$, the sufficient condition for stabilization
becomes 
\begin{align}
C \ge \log a \, (\gmm) \;\;\;[\mbox{\rm{nats}}/\mbox{\rm{sec}}],
\end{align}
where $C$ is the   timing capacity of the discrete telephone signaling channel~\cite{bedekar1998information}. The timing capacity is achieved in this case
 using i.i.d.\ waiting times $\{W_i\}$ that are distributed according to a mixture of a geometric and a delta distribution. This results in   $\{D_i\}$ also being i.i.d. geometric~\cite{sundaresan2000robust,bedekar1998information}.

Assuming that a decoding operation occurs at time $m$ using all $k_m$ symbols  received up to this time, and following the source-channel coding scheme described in the proof of Theorem~\ref{SUFFCONDMA},  the controller decodes an estimate ${\hat{X}_m[0]}$ of the initial state  and estimates the current state as
\begin{align} \label{estcorrect}
    \XCH[m] &= 
    a^m\hat{X}_m[0]+\sum_{j=0}^{m-1} 
    {a^{m-1-j}U[j]}.
  \end{align}
  The  estimate ${\hat{X}_m[0]}$   corresponds to the binary representation of $X(0)$ using $
 \lceil k_m  \mathbb{E}(D) C 
\rceil$ bits,  provided that there is no decoding error in the tranmsission. Accordingly,   in our simulation,  we let $\eta>0$ and $P_{e}=e^{-\eta k_m}$, and we assume that   at every decoding time,    with probability $(1-P_{e})$ we construct a correct quantized estimate of the initial state ${\hat{X}_m[0]}$ using   $\lceil k_m \mathbb{E}(D) C \rceil$  bits. 
  Alternatively,   with probability $P_e$ we construct an  incorrect  quantized estimate. 
In the case of a correct estimate, we apply the asymptotically optimal control input $U[m]= - K \XCH[m]$, where $K>0$ is the control gain and $\XCH[m]$ is obtained from \eqref{estcorrect}. In the case of an incorrect estimate, the  state estimate
used to construct the control input can be arbitrary. We consider three cases: (i) we do not apply any control input and let the system evolve in open loop, (ii) we apply the control input using the previous estimate, (iii) we apply the opposite of the asymptotically optimal control input: $U[m]= K \XCH[m]$.   In all cases, the control input remains fixed to its most recent value during the time required for a new estimate to be performed.

  These three cases are included as representative sensitivity scenarios: case (i) serves as a conservative open-loop fallback, case (ii) as a natural hold-last-estimate surrogate, and case (iii) as an intentionally adverse benchmark. Among these, case (ii) is the most practically interpretable if an external error-detection mechanism is available, whereas case (iii) should be viewed only as a stress test.
  
  In this case study, we set 
$K=0.4$, chosen to minimize an linear quadratic regulator (LQR) cost as reported in the caption of Fig.~\ref{simres}. This choice is used only to define a nominal stabilizing controller for the simulation; the information-theoretic results of the previous sections do not depend on this particular contorl policy.

  \begin{figure}[!thb]
	\centering
 \includegraphics[width=.41 \textwidth]{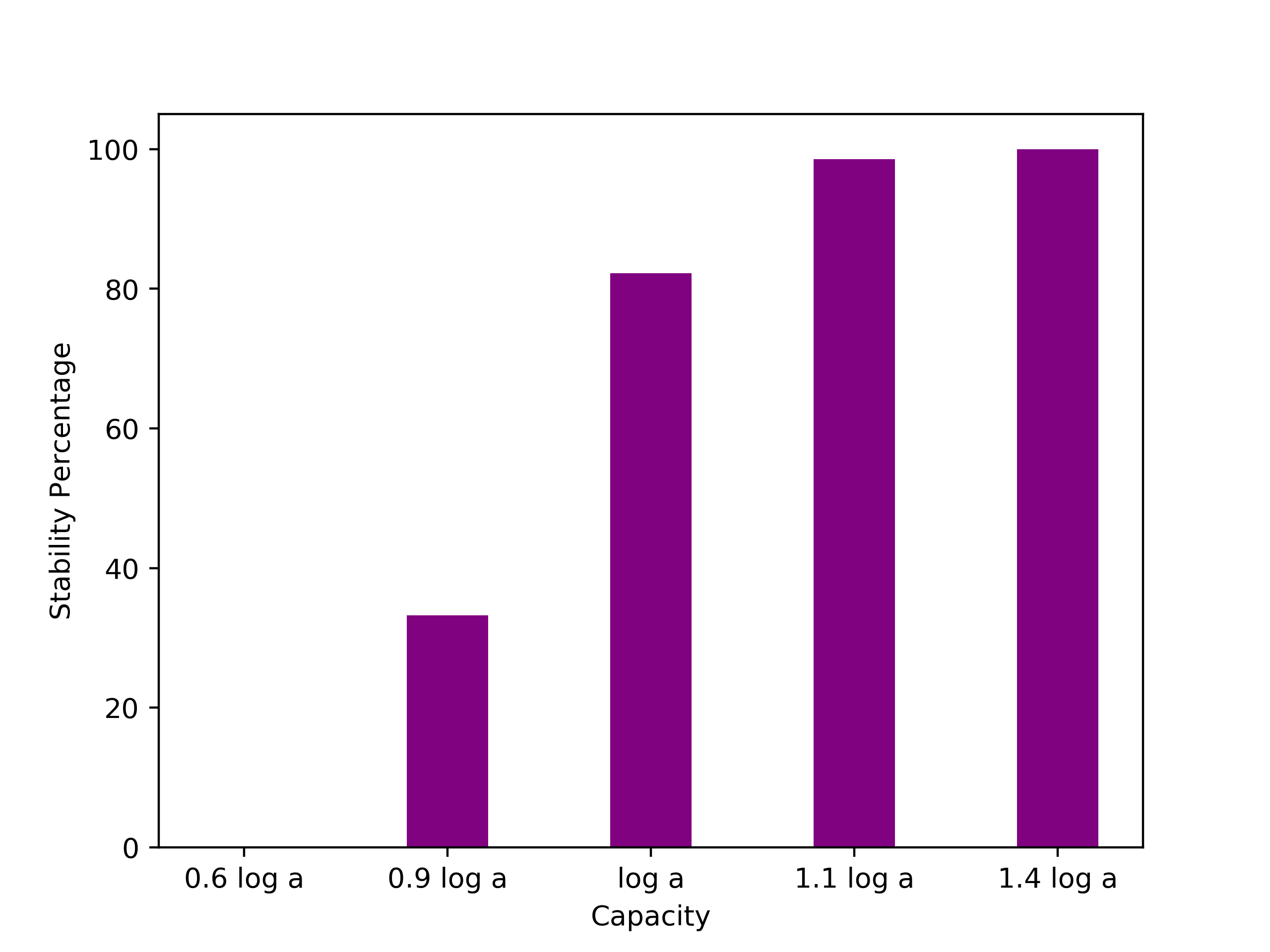}
	\caption{ \footnotesize{Here we show the fraction of times stabilization was achieved  versus the capacity of the channel across a run of 500 simulations for each value of the capacity.
 Successful stabilization is defined in these simulations as $|X[250]|\le 0.05$. In the case of a decoding error, no control input is applied and we  let the system evolve in open loop.
The  simulation parameters were chosen as follows: $a=1.2$, $\mathbb{E}(D)=2$, $P_{e} = e^{-\eta k_m}$, where $\eta=0.09$, and the control gain is $K=0.4$. 
 }
 }\label{fig:barchart_sucessive}
\end{figure}

Fig.~\ref{timing_simulations3} pictorially illustrates the evolution of our simulation in an error-free case in which the  binary representation of $X[0]$ is refined by $\mathbb{E}(D) C=3$ bits at each symbol reception.

Numerical results are depicted in  Fig.~\ref{simres}, 
\begin{figure*}[t]
\centering
\begin{tabular}{c c c}
    &\scriptsize{Case I: decoding error $\rightarrow$ open loop}\\
    	\scriptsize{$C=1.2 \log a$} &
    \scriptsize{$C=1.2 \log a$}
    & \scriptsize{$C=0.9 \log a$}
    \\
    \includegraphics[width=40mm]{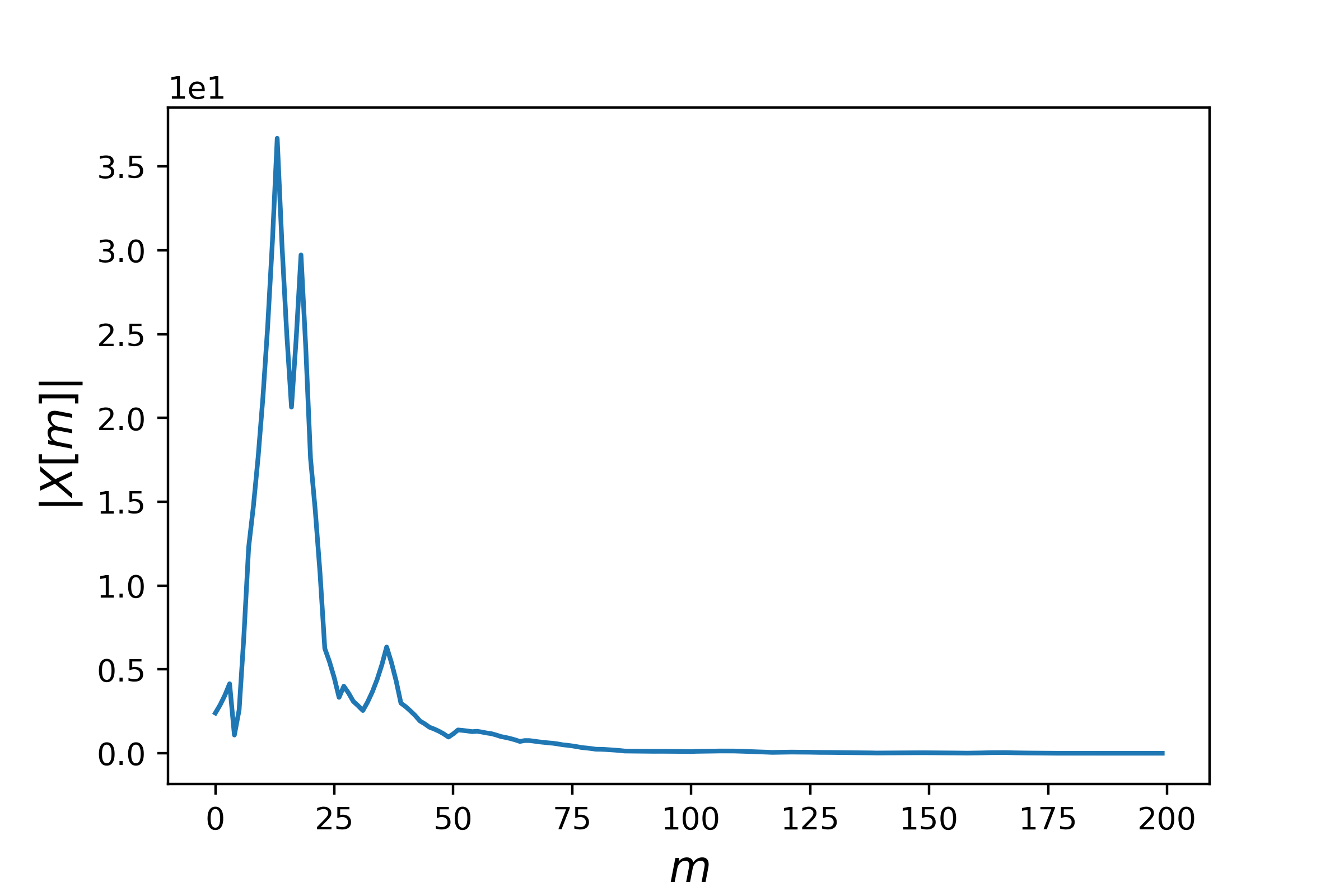} &
    \includegraphics[width=40mm]{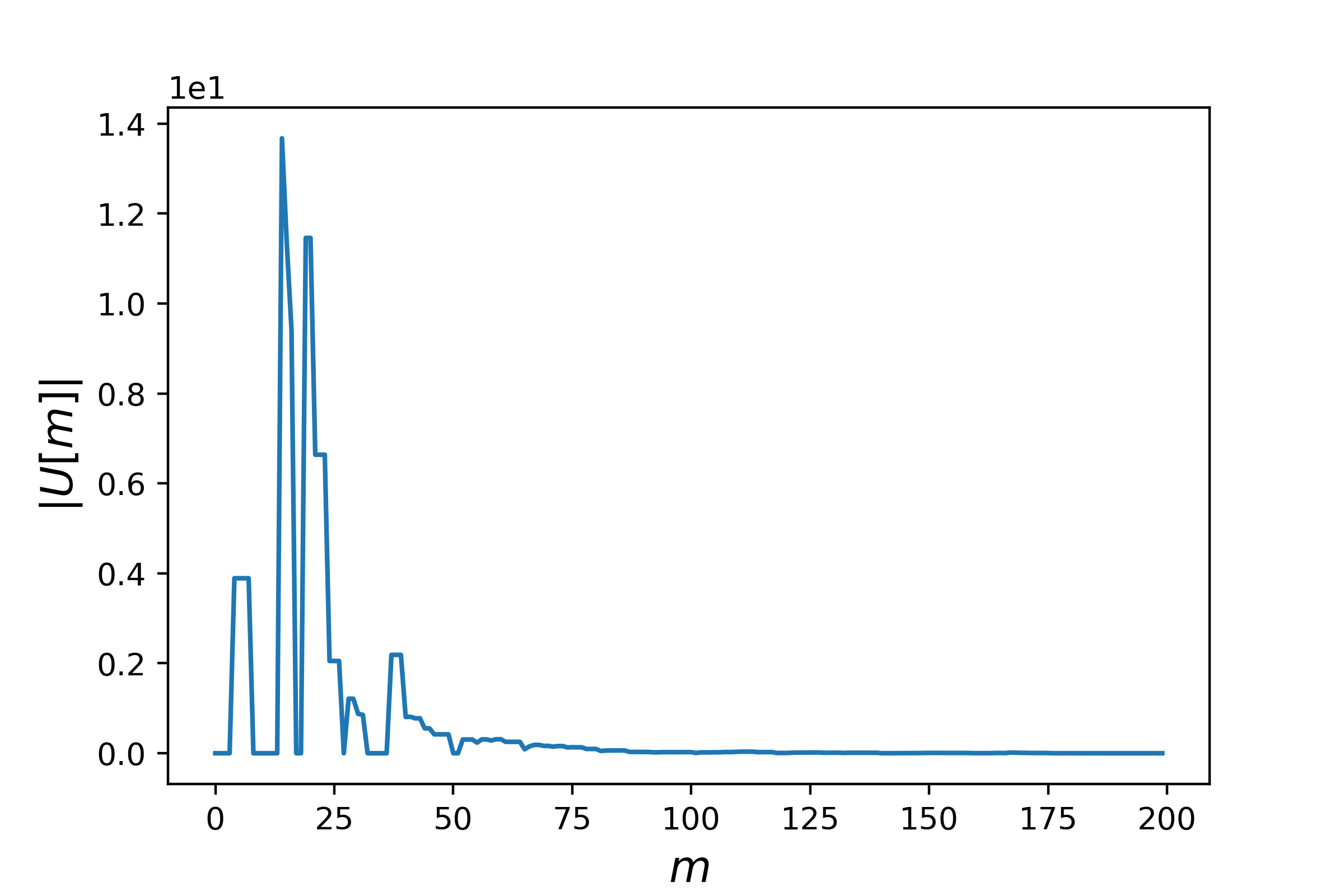} &
    \includegraphics[width=40mm]{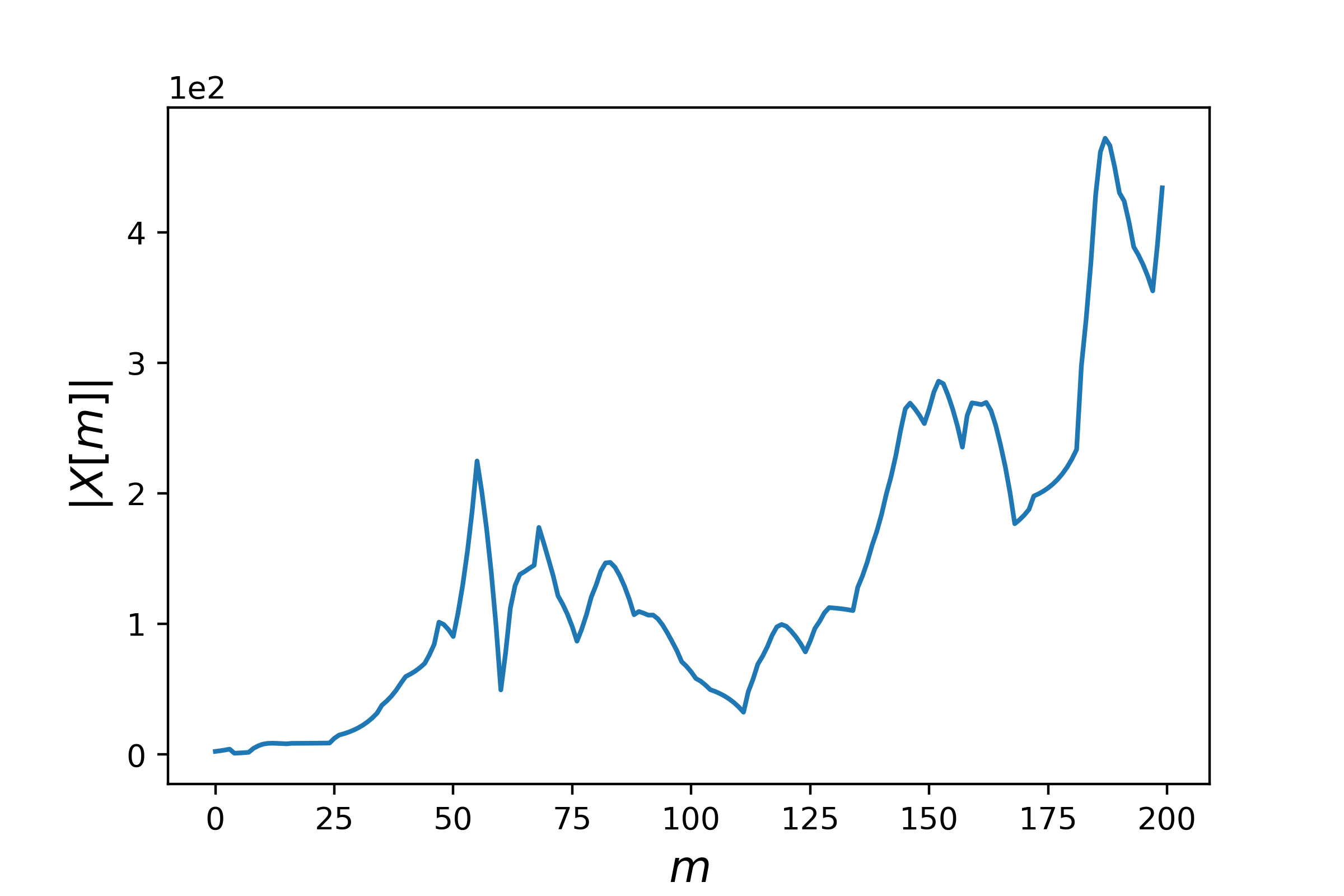} 
    \\
     &\scriptsize{Case II: decoding error $\rightarrow$ previous estimate}\\
    	\scriptsize{$C=1.2 \log a$} &
    \scriptsize{$C=1.2 \log a$} 
     & \scriptsize{$C=0.9 \log a$}
    \\
    \includegraphics[width=40mm]{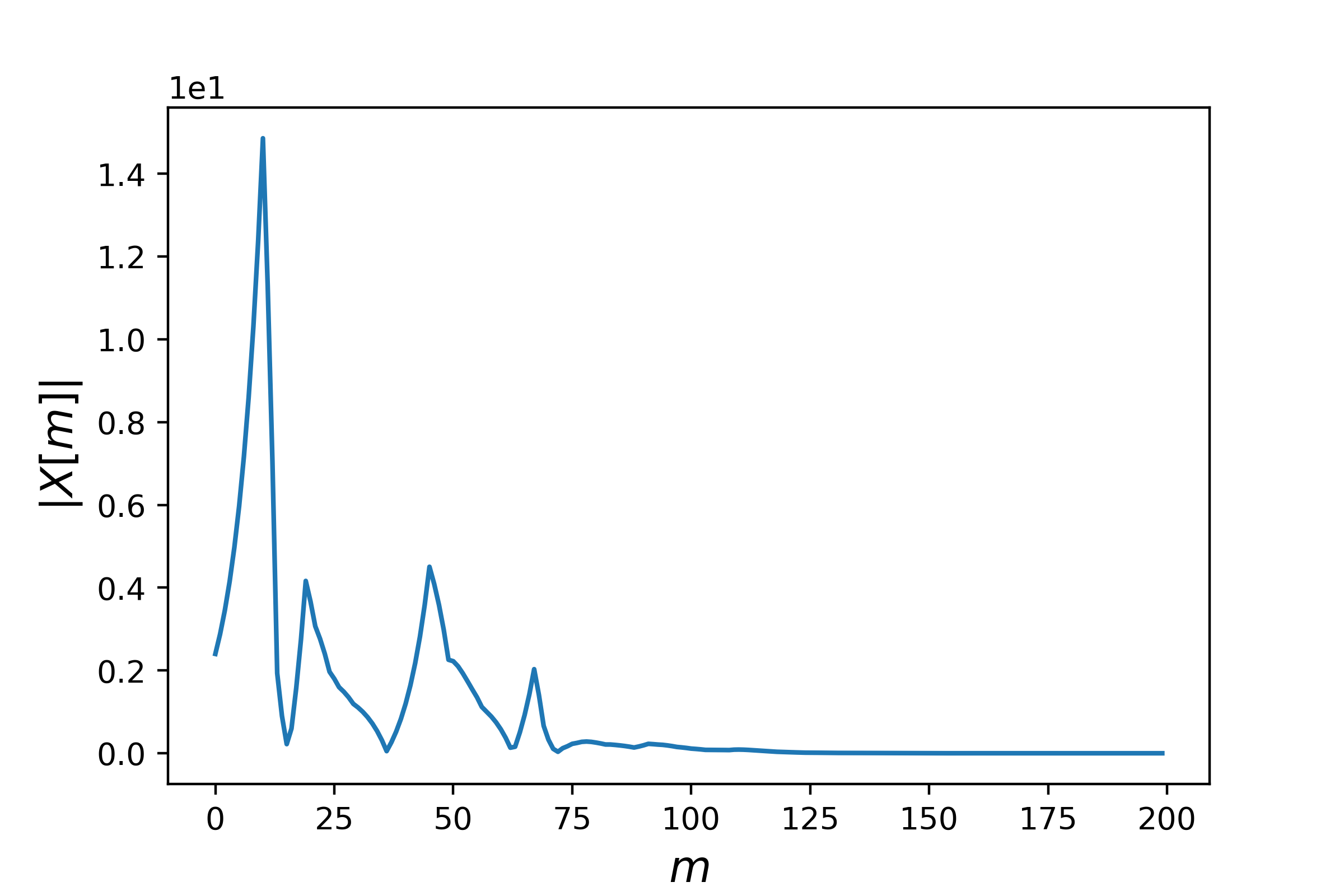} &
    \includegraphics[width=40mm]{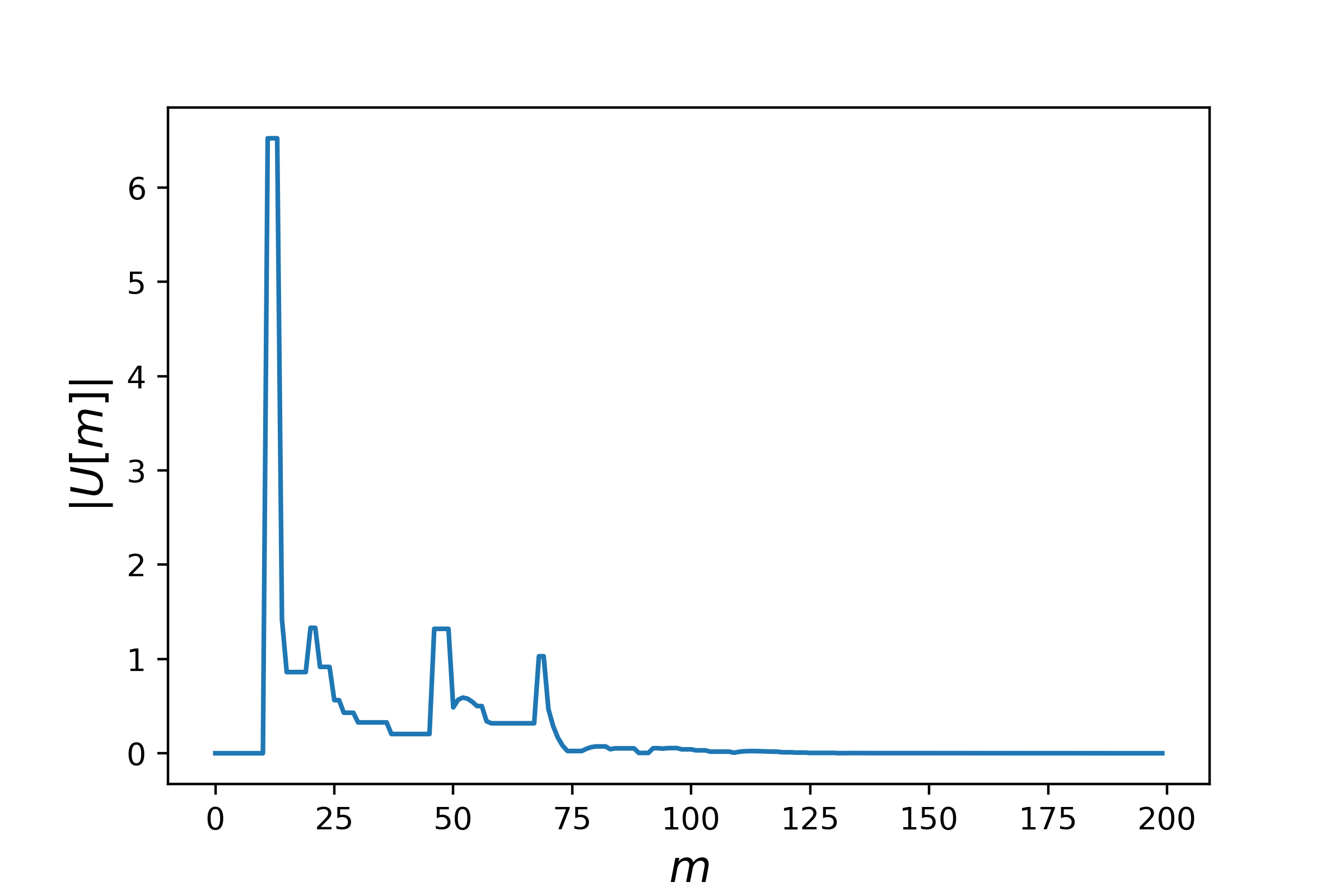} &
    \includegraphics[width=40mm]{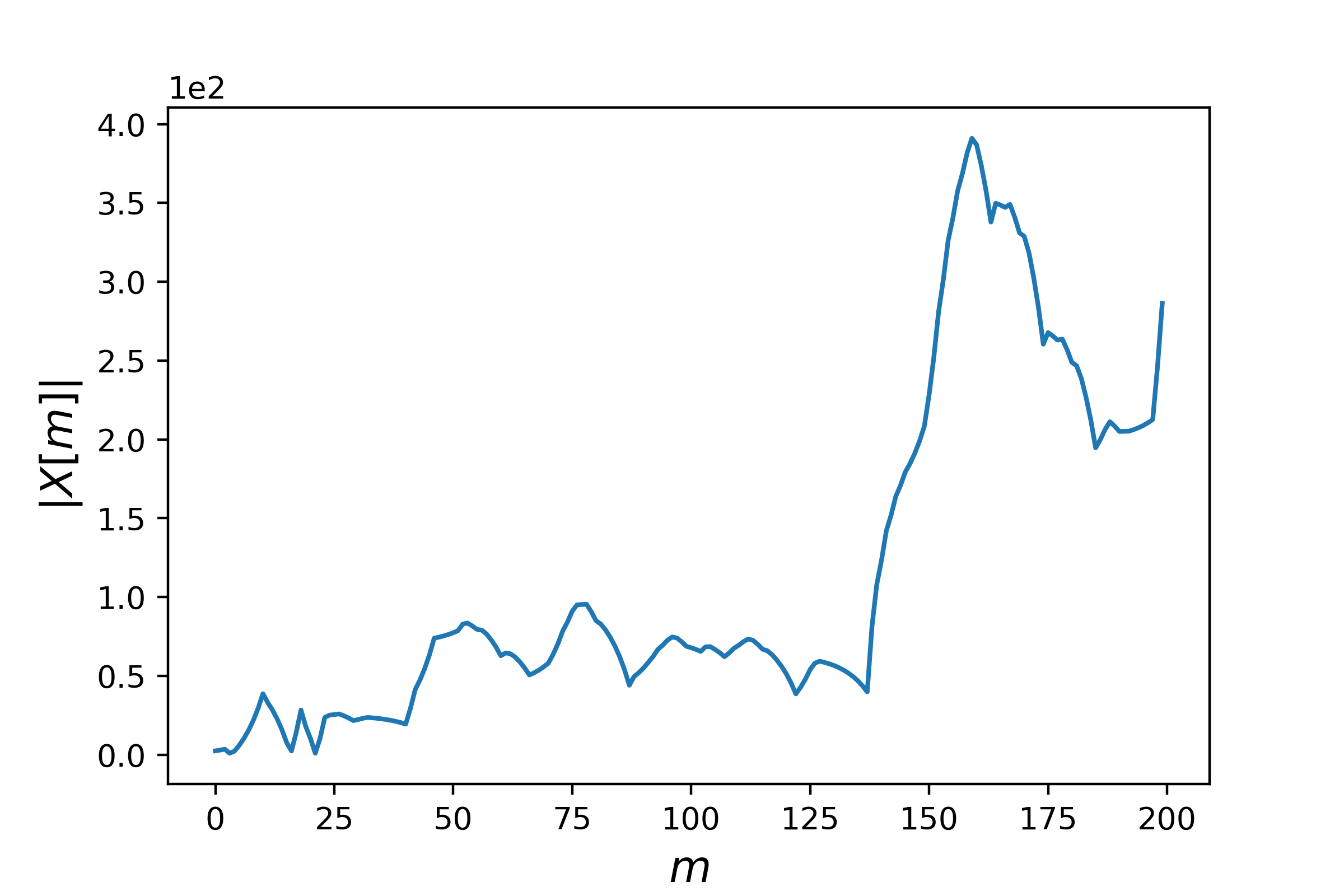} 
      \\
       &\scriptsize{Case III: decoding error $\rightarrow$ opposite of the optimal control}\\
    	\scriptsize{$C=1.2 \log a$} &
    \scriptsize{$C=1.2 \log a$} 
     & \scriptsize{$C=0.9 \log a$}
    \\
    \includegraphics[width=40mm]{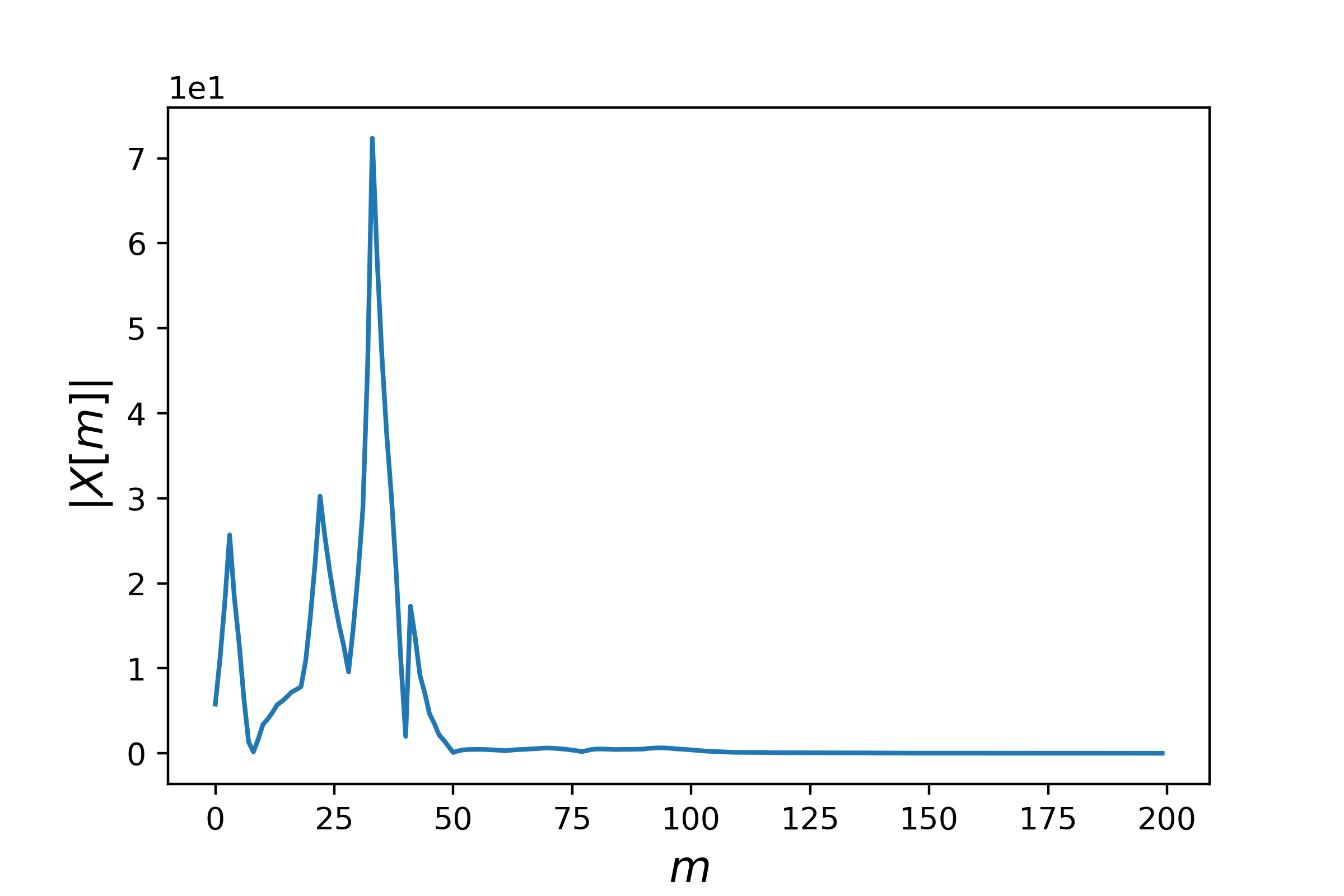} &
    \includegraphics[width=40mm]{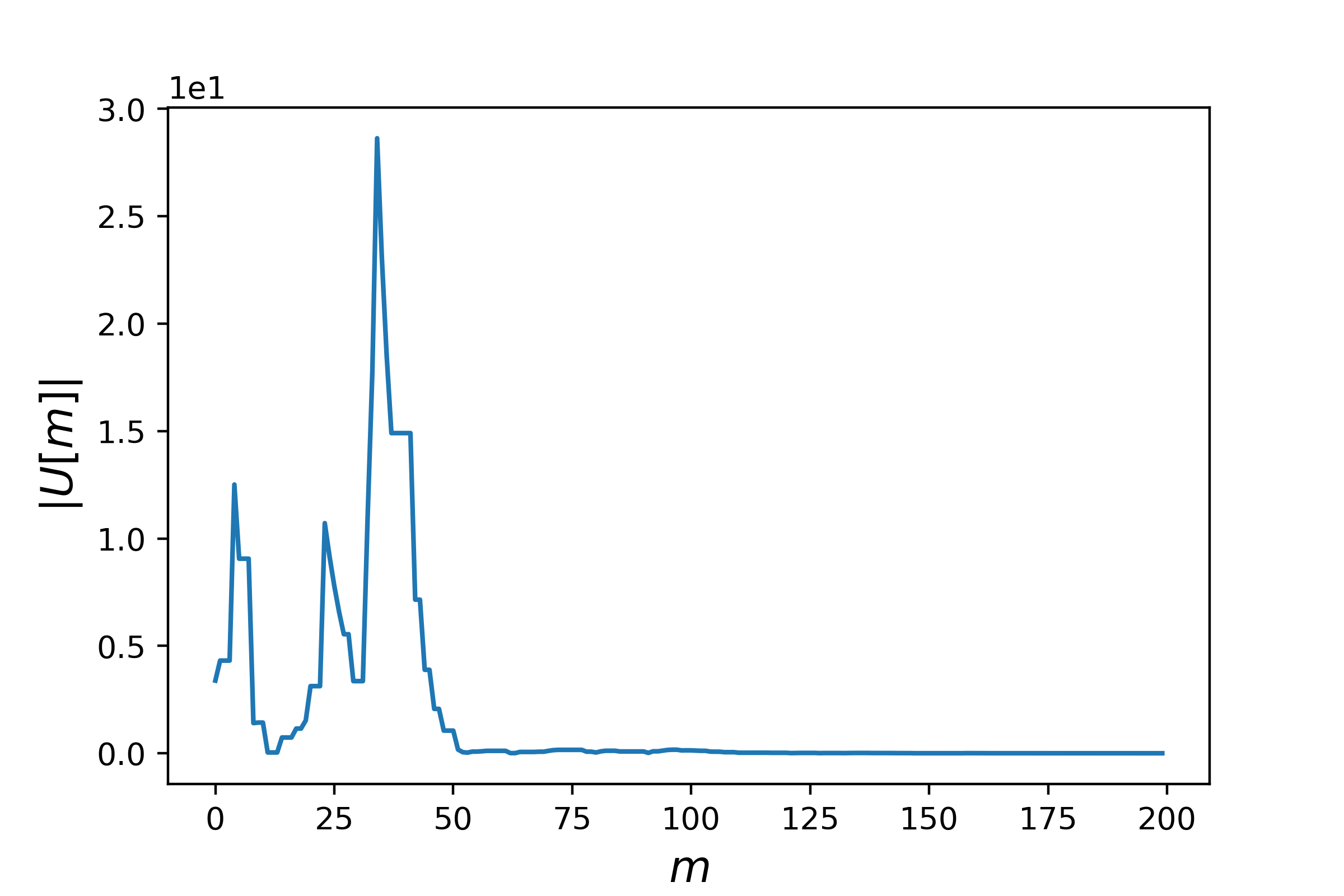} &
    \includegraphics[width=40mm]{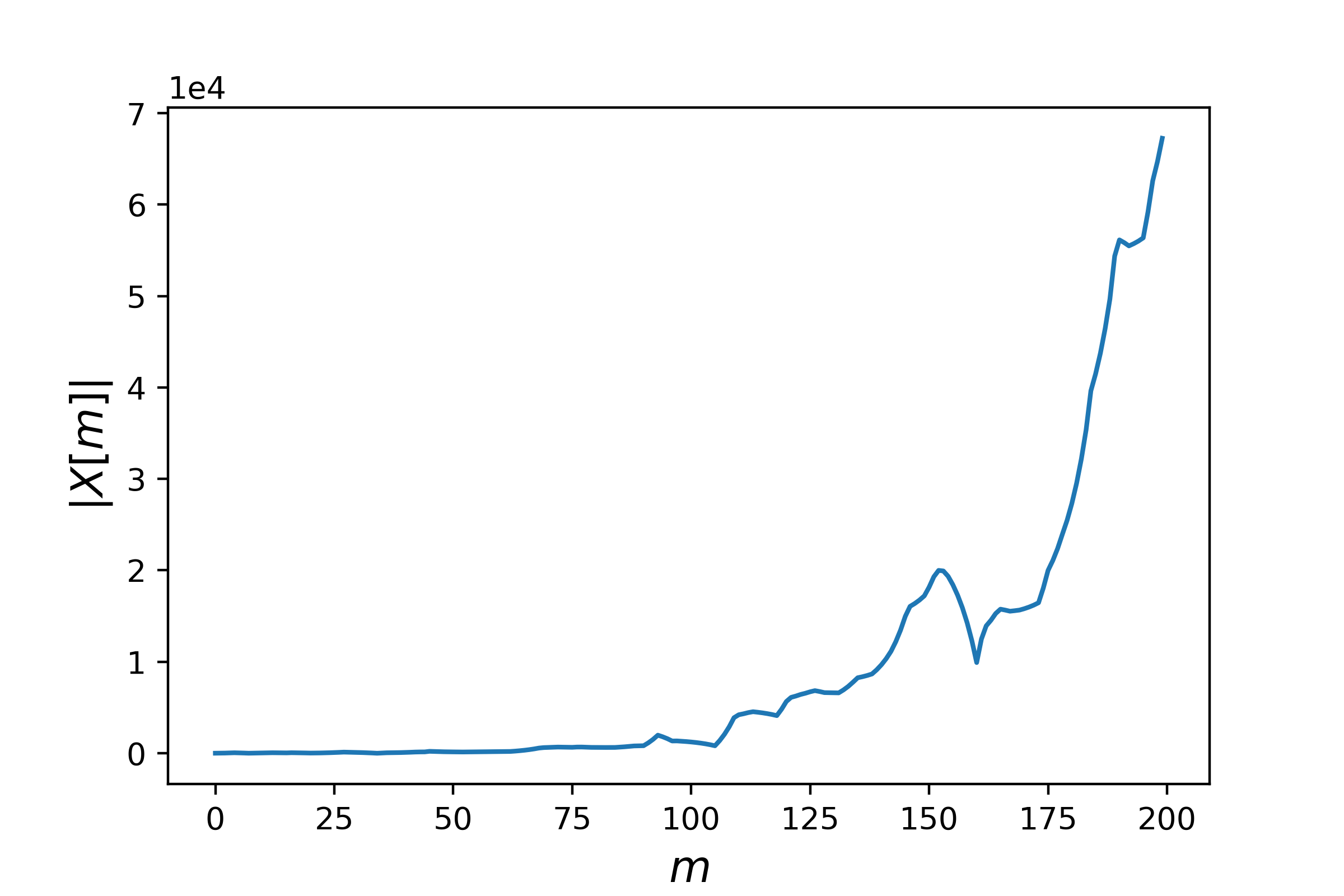} 
\end{tabular}
\caption{
Here we show the evolution of a single run of a system with different capacities for the timing channel.
The first and second columns represent the  absolute value of the  state and control input, respectively,  when the timing capacity is larger than the entropy rate of the system $(C>\log a)$. The third column represents the  absolute value of the  state when the timing capacity is smaller than the entropy rate of the system $(C<\log a)$. 
In the first row, in the presence of a decoding error, we do not apply any control input and let the system evolve in open-loop; 
in the second row,  we apply the control using the previous estimate;
the third row,    we apply the opposite of the optimal control.  
The simulation parameters were chosen as follows: $a=1.2$, $\mathbb{E}(D)=2$, and $P_{e} = e^{-\eta k_m}$, where $\eta=0.09$. 
 For the optimal control gain we have chosen $K=0.4$, which is optimal with respect to the (time-averaged) linear quadratic regulator (LQR) control cost $(1/200)\mathbb{E}[\sum_{m=0}^{199} (0.01X_k^2+0.5U_k^2)+0.01X_{200}^2].$
}
\label{simres}
\end{figure*}
showing   convergence of the state to zero in all cases, provided that the timing capacity is above the entropy rate of the system. In contrast, when the timing capacity is below the entropy rate, the state diverges. The plots also show the  absolute  value of the control input    used for stabilization in the various cases.

Fig.~\ref{fig:barchart_sucessive} illustrates the percentage of times at which the controller successfully stabilized the plant versus the capacity of the channel in a run of  500 Monte Carlo simulations. The phase transition behavior at the critical value $C=\log a $ is clearly evident.

\begin{remark}
The numerical example in this section is not intended as a direct numerical discretization study of the continuous-time system in~\eqref{syscon}. Rather, it serves as a tractable discrete-time surrogate model that illustrates the core information-theoretic mechanism developed in the preceding sections. The theoretical results of this paper are formulated in continuous time because the plant dynamics, the timing-capacity characterization, and the comparison with event-triggered formulations~\cite{OurJournal1} are naturally expressed on a continuous-time axis.

Assuming a sampling period $\Delta$, if $a_c$ denotes the open-loop gain of the continuous-time system, then the corresponding discrete-time gain is $a_d=e^{a_c\Delta}$. Hence, the open-loop uncertainty growth in discrete time is $\ln a_d=a_c\Delta$ per step, i.e., $(\ln a_d)/\Delta=a_c$ per unit time. Thus, while discretization changes the numerical values when expressed per step, it does not change the underlying per-unit-time information balance that motivates the theoretical results. In this sense, the numerical example should be viewed as a tractable illustration of the information-theoretic mechanism developed in the preceding sections.
\end{remark}

\section{Conclusion and outlook}
\label{sec:conc}
In the framework of control of dynamical systems over communication channels, it has recently been observed that  event-triggering policies  encoding information over time in a state-dependent fashion  can exploit timing information for stabilization   in addition to the information traditionally carried by data packets~\cite{khojasteh2020exploiting,Level,OurJournal1,linsenmayer2017delay,yildiz2019event,cdc19papermj,guo2019optimal}.  In a more general framework, this paper studied from an information-theoretic perspective the fundamental limitation of using \emph{only} timing information for stabilization, independent of any transmission strategy.
We showed that for stabilization of an undisturbed scalar linear system over a channel with a unitary alphabet, 
the timing capacity~\cite{anantharam1996bits} should be, essentially, at least as large as the entropy
rate of the system. In addition, in the case of exponentially
distributed delays, we provided an almost tight sufficient condition using a coding strategy that refines the estimate of the decoded message as more and more symbols are received.

Our derivation ensures that when the timing capacity is larger than the entropy rate,    the estimation error   does not grow unbounded, in probability, even in the presence of the random delays occurring in the timing channel. This is made possible by   communicating a real-valued variable (the initial state) at an increasingly higher resolution and with vanishing probability of error. This strategy   has been previously studied in~\cite{como2010anytime} in the context of estimation over the binary erasure channel, rather than over   the timing channel. It is also related to communication at increasing resolution over channels with feedback via posterior matching~\cite{shayevitz2011optimal,naghshvar2015extrinsic}. The classic Horstein~\cite{horstein1963sequential} and Schalkwijk-Kailath~\cite{schalkwijk1966coding} schemes   are special cases of posterior matching for the binary symmetric channel and the additive Gaussian channel respectively.
The main idea in our setting is to employ a tree-structured quantizer in conjunction to a capacity-achieving timing channel codebook that grows exponentially with the tree depth, and re-compute the estimate of the real-valued variable  as more and more channel symbols are received. The estimate is re-computed  for  a number of received symbols that depends on  the channel rate and on the average delay. In contrast to posterior matching, we are not concerned with the complexity of the encoding-decoding strategy, but only with its existence. We also do not assume a specific distribution for the real value we need to communicate, and we do not use the feedback signal to perform encoding, but only to avoid queuing~\cite{anantharam1996bits,sundaresan2000robust}. We  point out that our control strategy does  not work in the presence of disturbances: in this case, one needs   to track a state that  depends not only on the initial condition, but also on the evolution of the disturbance. This requires to  update  the entire   history of the system's states at each symbol reception~\cite{sahai2006necessity },  leading to a different, i.e. non-classical, coding model. Alternatively, remaining in a classical setting one could aim for less, and attempt to obtain results using  weaker probabilistic notions of stability, such as  the one in~\cite[Chapter~8]{matveev2009estimation}.

Finally, by  showing that in the case of no disturbances and exponentially distributed delay  it is possible to achieve stabilization at zero data-rate  only for sufficiently small average delay $\mathbb{E}(S)<(e \, a)^{-1}$, we confirmed  from an information-theoretic perspective the observation made in~\cite{OurJournal1} regarding the existence of a critical delay value for stabilization at zero data-rate.

We conclude by discussing future directions motivated by our analysis.

\noindent\textbf{Extensions to vector systems:}
The present analysis focuses on scalar systems to keep the exposition concise. The fundamental limits developed in this work can be applied to vector systems with only one unstable mode  (cf.~\cite{khojasteh2020exploiting}). In general, for vector systems, classical data-rate theorems state that the information-generation rate is governed by the intrinsic entropy rate of the unstable dynamics (e.g., the sum of the positive real parts of the unstable eigenvalues in continuous time). We expect that the necessary condition in this paper admits an analogous extension based on the the intrinsic entropy rate of the system. However, establishing this rigorously would require extending the main steps of our scalar converse proof to the vector setting. Extending the sufficient results to vector plants is more challenging. Such an extension would require a timing-based estimation mechanism that can  reduce uncertainty across multiple unstable modes under random reception times, in a manner closer in spirit to encoder/decoder architectures studied in~\cite{Mitter} for noiseless digital channels. As in the scalar case, a key difficulty is to achieve convergence in probability through an anytime-reliable transmission mechanism despite unbounded random delays in the timing channel, which requires careful analysis.

\noindent\textbf{Effect of system disturbances:}
As discussed above, the present study considers a disturbance-free plant, with the objective of driving the state to zero, with probability one, using information encoded solely in the inter-reception times of symbols. When disturbances are present in system dynamics, convergence to the origin needs to be replaced by alternative notions such as probabilistic ultimate boundedness or moment stability. In these settings, the information-theoretic quantity governing stabilizability may require stronger reliability notions than Shannon capacity; for example, for discrete memoryless channels, moment objectives are closely related to anytime capacity, whereas almost-sure boundedness under bounded disturbances naturally invokes zero-error requirements (cf. Table~\ref{tab:table1}). Developing timing-channel counterparts of these notions, defined in terms of inter-reception times rather than data payloads, is an important and technically challenging direction for future work.

\noindent\textbf{Understanding the combination of timing information and packets with data payload:}
As the controller may exploit both payload and timing information, a more general data-rate formulation should account for two distinct information flows. Developing fundamental stabilization results that explicitly quantify the trade-off between payload capacity and timing capacity is therefore an important direction for future research in networked control systems.

\noindent\textbf{Other delay models:} 
In this work, we focus on exponentially distributed delays, which are not only analytically convenient, but also constitute a canonical queue-based timing-channel model~\cite{sundaresan2006capacity}. Prior work has also considered non-exponential service-time laws, including uniform and truncated-Gaussian timing-channel models~\cite{sellke2006timing,sellke2007capacity}. Although the entropy-rate viewpoint underlying our necessary condition is expected to remain relevant more generally, extending the present results to non-exponential delay models is left for future work.

\noindent\textbf{Complexity-performance trade-offs:} Our proofs are information-theoretic and focus on the existence of encoding/decoding strategies rather than their computational complexity. Once a codebook and quantizer are fixed, the encoder can be implemented via a simple lookup that selects the prescribed waiting time and uses feedback only to avoid queuing. The main computational burden lies at the decoder, since the capacity-achieving existence argument relies on random codebooks and is not designed for real-time implementation~\cite{gallager2003low}. Developing low-complexity timing-channel codes and tractable sequential anytime decoders (cf.~\cite{sundaresan2002sequential}), together with quantitative performance-complexity trade-offs, is an important direction for future work.

Collectively, these directions suggest that timing information will remain central to the fundamental limits of networked control systems, with many opportunities for further theoretical and practical developments.

\section*{APPENDIX}
\begin{figure*}[t]
	\centering
 \includegraphics[scale=0.55]
 {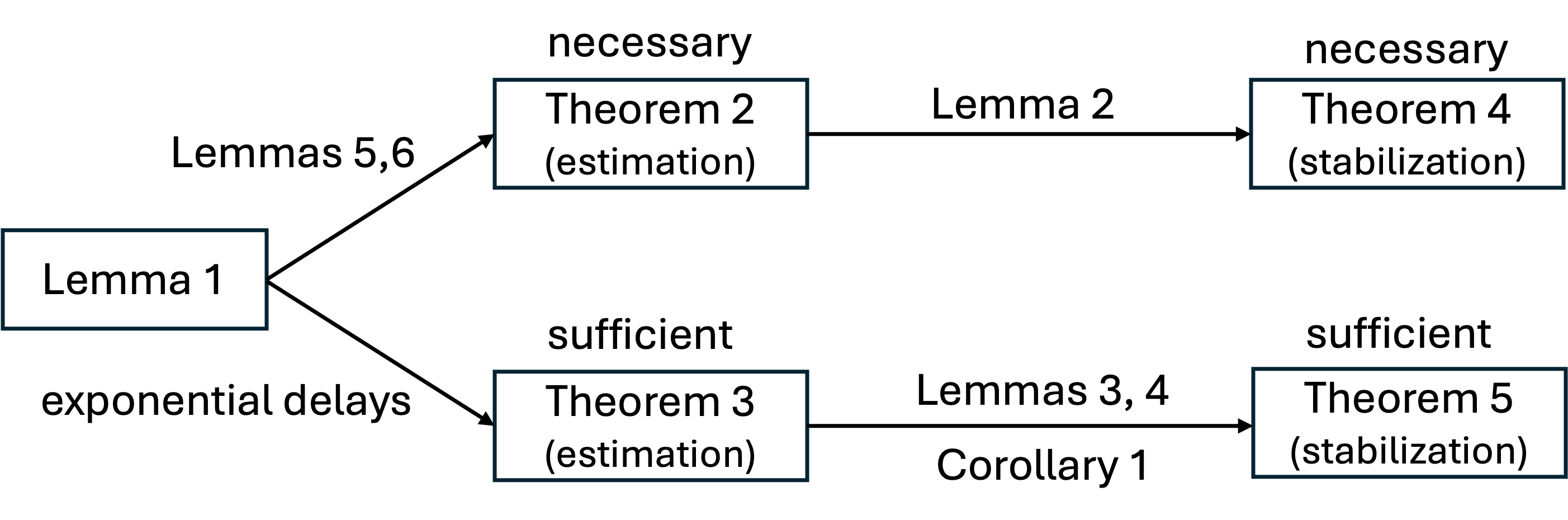}
\caption{\footnotesize{Roadmap of the main proofs.}} \label{diagram-1}
\end{figure*}
\subsection{Roadmap of the Proofs}\label{flowchart-1}
Fig.~\ref{diagram-1} provides a schematic roadmap of the logical dependencies among the main lemmas and theorems underlying the necessary and sufficient conditions.

\subsection{Proof of Theorem~\ref{NECCONDITIO}}\label{Necesssary}
We start by introducing a  few definitions and proving some useful lemmas.

\begin{defn}
\label{ratedistortion}
For any $\epsilon >0$ and $\phi>0$, we define  the rate-distortion function of  the  source $\dot{\XE}=a\XE(t)$ at times $\{\timenN\}$ as
\begin{small}
\begin{align}\label{defratedisfunc}
R_{\timenN}^\epsilon(\phi)
= \inf_{\mathbb{P}\left(\XEH(\timenN)|\XE(\timenN)\right)}\Bigg\{I\left(\XE(\timenN);\XEH(\timenN)\right) :
\\\nonumber
 \mathbb{P}\left(|\XE(\timenN)-\XEH(\timenN)|>\epsilon\right)\le \phi\Bigg\}.
\end{align}
\end{small}
\end{defn}

The proof of the following lemma adapts an argument of~\cite{tatikonda2004control}   to our continuous-time setting.
\begin{lemma}\label{ratedistortionlower}
We have
\begin{align}\label{lowerratedis}
R_{\timenN}^\epsilon(\phi) \ge (1-\phi) \left[a\,\timenN+h(X(0))\right] \\\nonumber -\ln 2\epsilon -  \frac{\ln 2}{2} \;\;\; [\mbox{\rm{nats}}] .
\end{align}
\end{lemma}
\begin{proof}
Let 
\begin{align}
\label{ep}
\xi=\left \{
  \begin{array}{cc}
  0 & \mbox{if}~|\XE(\timenN)-\XEH(\timenN)|\le \epsilon \\
  1 & \mbox{if}~|\XE(\timenN)-\XEH(\timenN)|> \epsilon.
  \end{array}
\right.
\end{align} 
Using the chain rule, we have 
\begin{align}\label{Khab1}
&I(\XE(\timenN);\XEH(\timenN))\nonumber
\\
&=I(\XE(\timenN);\xi,\XEH(\timenN))-I(\XE(\timenN);\xi|\XEH(\timenN))
 \\
&=I(\XE(\timenN);\xi,\XEH(\timenN))-H(\xi|\XEH(\timenN))\nonumber
\\\nonumber
&~~+H(\xi|\XE(\timenN),\XEH(\timenN)).\nonumber
\end{align} 
Given $X(\timenN)$ and $\hat{X}(\timenN)$, there is no uncertainty in $\xi$,  hence we deduce
 \begin{align}\label{khab2}
&I(\XE(\timenN);\XEH(\timenN))\nonumber \\
&=I(\XE(\timenN);\xi,\XEH(\timenN))-H(\xi|\XEH(\timenN))\nonumber\\
&=h(\XE(\timenN))-h(\XE(\timenN)|\xi,\XEH(\timenN))-H(\xi|\XEH(\timenN))\nonumber\\
&=h(\XE(\timenN))-h(\XE(\timenN)|\xi=0,\XEH(\timenN))\mathbb{P}(\xi=0)
\\
&~~-h(\XE(\timenN)|\xi=1,\XEH(\timenN))\mathbb{P}(\xi=1)
-H(\xi|\XEH(\timenN)).\nonumber
\end{align}
Since ${H(\xi|\XEH(\timenN)) \le H(\xi) \le \ln 2/2~[\mbox{nats}]}$, ${\mathbb{P}(\xi=0) \le 1}$, and ${\mathbb{P}(\xi=1) \le \phi}$, it then follows that
 \begin{align}
&I(\XE(\timenN);\XEH(\timenN)) \ge
\nonumber \\
&
h\left(\XE(\timenN)\right)
-h\left(\XE(\timenN)-\XEH(\timenN)|\xi=0,\XEH(\timenN)\right)\nonumber \\
&~~-h\left(\XE(\timenN)|\xi=1,\XEH(\timenN)\right)\phi-\frac{\ln 2}{2}.
\end{align}
Since conditioning reduces the entropy, we have
 \begin{align}
&I(\XE(\timenN);\XEH(\timenN))\ge h(\XE(\timenN))
\\\nonumber
&-h(\XE(\timenN)-\XEH(\timenN)|\xi=0)-h(\XE(\timenN))\phi-\frac{\ln 2}{2}\\\nonumber
&=(1-\phi)h(\XE(\timenN))-h(\XE(\timenN)-\XEH(\timenN)|\xi=0)-\frac{\ln 2}{2}.
\end{align}
By \eqref{ep} and since the uniform distribution maximizes the differential entropy among all distributions with bounded support, we have
 \begin{align}\label{khab4}
I(\XE(\timenN);\XEH(\timenN)) \ge (1-\phi)h(\XE(\timenN))-\ln 2\epsilon-\frac{\ln 2}{2}.
\end{align}
Since $\XE(\timenN)=X(0)~e^{a\timenN}$, we have
\begin{align}\label{khab5}
h(\XE(\timenN))&=\ln e^{a\timenN}+h(X(0))
=a\timenN+h(X(0)).
\end{align}
Combining~\eqref{khab4}, and~\eqref{khab5} we obtain
 \begin{align}\label{lowerratedis11}
I(\XE(\timenN);\XEH(\timenN))\ge 
(1-\phi) \left(a\timenN+h(X(0))\right)-\ln 2\epsilon -\frac{\ln 2}{2}.
\end{align}
Finally, noting that this inequality is independent of $\mathbb{P}(\XEH(\timenN)|\XE(\timenN))$ the result follows.
\end{proof}
\hfill \break
\begin{remark}
By letting  $\phi=\epsilon$ in~\eqref{lowerratedis}, we have
\begin{align}
R_{\timenN}^\epsilon(\epsilon) \ge (1-\epsilon) a\timenN+\epsilon',
\end{align}
where
\begin{align}
\epsilon'=(1-\epsilon)h\left(X(0)\right)-\ln 2\epsilon -\frac{\ln 2}{2}.
\end{align}
For sufficiently small $\epsilon$ we have that $\epsilon' \geq 0$, and hence
\begin{align}
\frac{R_{\timenN}^\epsilon(\epsilon) }{\timenN}\ge (1-\epsilon)a.
\end{align}
It follows that for sufficiently small $\epsilon$ the rate-distortion per unit time of the source must be at least as large as the entropy rate of the system.  Since the rate-distortion represents the number of bits required to represent the state of the process up to a given fidelity, this provides an operational characterization of the entropy rate of the system. \oprocend
\end{remark}

The proof of the following lemma follows a converse argument of~\cite{anantharam1996bits} with some modifications due to our different setting.
\begin{lemma}\label{mutalinfotimin}
Under the same assumptions as in Theorem~\ref{NECCONDITIO}, 
 if by time $\timenN$, $\kappa_n$ symbol is received by the controller,
we have 
\begin{align}
I\left(\XE(\timenN);\XEH(\timenN)\right) \le %\sum_{i=1}^{n}
\kappa_n I(W;W+S).
%~~\mbox{a.s.}
\end{align}
\end{lemma}
\begin{proof}
We denote the transmitted message by $V \in \{1,\ldots,M\}$ and the decoded message by $U \in \{1,\ldots,M\}$. Then 
\begin{align}
\XE(\timenN)\rightarrow V \rightarrow (D_1,\ldots,D_{\kappa_n}) \rightarrow U \rightarrow \XEH(\timenN),
\end{align}
is a Markov chain.
Therefore, using the data-processing inequality~\cite{cover2012elements}, we have
\begin{align}\label{eqmida5}
I\left(\XE(\timenN);\XEH(\timenN)\right) \le I(V;U) \le  I (V;D_1,\ldots,D_{\kappa_n}).
\end{align}
By the chain rule for the mutual information, we have 
\begin{align}\label{eqmida4}
I(V;D_1,\ldots,D_{\kappa_n})=\sum_{i=1}^{{\kappa_n}} I(V;D_i|D^{i-1}).
\end{align}
Since $W_i$ is uniquely determined by the encoder from $V$, using the chain rule we deduce
\begin{align}\label{SIA}
\sum_{i=1}^{\kappa_n} I(V;D_i|D^{i-1})=\sum_{i=1}^{\kappa_n} I(V,W_i;D_i|D^{i-1}).
\end{align}
In addition, again using the chain rule, we have
\begin{align}\label{partion2}
\sum_{i=1}^{\kappa_n} I(V,W_i;D_i|D^{i-1})&=\sum_{i=1}^{\kappa_n} I(W_i;D_i|D^{i-1})
\\\nonumber
&+\sum_{i=1}^{\kappa_n} I(V;D_i|D^{i-1},W_i).
\end{align}
$D_i$ is conditionally independent of $V$ when given $W_i$, hence, 
\begin{align}\label{SIA2}
\sum_{i=1}^{\kappa_n} I(V;D_i|D^{i-1},W_i)=0.
\end{align}
Combining~\eqref{SIA},~\eqref{partion2}, and~\eqref{SIA2} it follows that
\begin{align}\label{eqmida3}
\sum_{i=1}^{\kappa_n} I(V;D_i|D^{i-1})=\sum_{i=1}^{\kappa_n} I(W_i;D_i|D^{i-1}).
\end{align}
Since the sequences $\{S_i\}$ and $\{W_i\}$ are  i.i.d. and  independent of each other, it follows that the sequence $\{D_i\}$ is also  i.i.d., and we have
\begin{align}\label{eqmida10}
\sum_{i=1}^{\kappa_n} I(W_i;D_i|D^{i-1})=\kappa_nI(W;D).
\end{align}
By combining ~\eqref{eqmida5},~\eqref{eqmida4},  \eqref{eqmida3} and~\eqref{eqmida10} the result follows.
\end{proof}

We are now ready to finish the proof of Theorem~\ref{NECCONDITIO}.
\begin{proof}
If $\expvalue(W+S)=0$,
~\eqref{mainresult} is straightforward. Thus, for the rest of the proof, we assume $\expvalue(W+S)>0$.  Using Lemma~\ref{lemnewsub20},
as $n \rightarrow \infty$, by time $\timenN$, given in~\eqref{tamoomnemishe3}, with  a probability that tends to one, at most $n$ symbols are received by the controller.
In this case, using Lemma~\eqref{mutalinfotimin},  it follows that
\begin{align}\label{GGMJ22}
 n \, I(W;W+S) \ge I\left(\XE(\timenN);\XEH(\timenN)\right).
\end{align}
By the assumption  of the theorem, for any $\epsilon>0$ we have
\begin{align}\label{helpful12}
\lim_{n\rightarrow\infty}\mathbb{P}\left(|\XE(\timenN)-\XEH(\timenN)|\le \epsilon\right)=1.
\end{align}
Hence, for any $\epsilon>0$  and any $\phi>0$ there exist $n_\phi$ such that for $n\ge n_\phi$
\begin{align}\label{IEE2}
\mathbb{P}\left(|\XE(\timenN)-\XEH(\timenN)|>\epsilon\right)\le \phi.
\end{align}
Using~\eqref{IEE2},~\eqref{defratedisfunc}, and Lemma~\ref{ratedistortionlower} it follows that for $n\ge n_\phi$
\begin{align}\label{shab2}
R_{\timenN}^\epsilon(\phi) \ge (1-\phi) \left[a\timenN+h(X(0))\right]-\ln 2\epsilon -\frac{\ln 2}{2}.
\end{align}
By~\eqref{defratedisfunc},  we have
\begin{align}\label{Icds23}
I(\XE(\timenN);\XEH(\timenN)) \ge R_{\timenN}^\epsilon(\phi),
\end{align}
and combining~\eqref{shab2}, and~\eqref{Icds23} we obtain that for $n\ge n_\phi$
\begin{align}
&\frac{I\left(\XE(\timenN);\XEH(\timenN)\right)}{n} \ge
\\\nonumber
&\frac{(1-\phi)a\timenN}{n}  
+\frac{(1-\phi) h(X(0))-\ln 2\epsilon -\frac{\ln 2}{2}}{n}.
\end{align}
We now let $\phi \rightarrow 0$, so that $n\rightarrow \infty$.
Using~\eqref{GGMJ22} we have
\begin{align}\label{Gfq22}
I(W;W+S) \ge a \lim_{n \rightarrow \infty}\frac{\timenN}{n}.
\end{align}
Since, $\expvalue(\mathcal{T}_n)=n\expvalue(D_n)$
 from~\eqref{tamoomnemishe3} it follows that 
\begin{align}\label{G3311}
\lim_{n \rightarrow \infty} \frac{\timenN}{n} \ge   (\ngmm) \, \expvalue(D).
\end{align}
Combining~\eqref{G3311} and~\eqref{Gfq22},~\eqref{mainresult} follows. 
Finally, using~\eqref{capanaverdum}
and noticing
\begin{align}\label{supineq8}
\sup_{\substack{W\ge 0 \\ \mathbb{E}(W)\le\chi}} \frac{I(W;W+S)}{\mathbb{E}(S)+\chi} \ge \sup_{\substack{W\ge 0 \\ \mathbb{E}(W)=\chi}} \frac{I(W;W+S)}{\mathbb{E}(S)+\chi},
\end{align}
we deduce that if~\eqref{mainresult} holds then \eqref{obs:cap:nec:12} holds as well.  
\end{proof}

\subsection{Proof of Theorem~\ref{SUFFCONDMA}}

\begin{proof}
If $\expvalue(S)=0$ the timing capacity is infinite, and the result is trivial. Hence, for the rest of the proof, we   assume that 
\begin{align}\label{sea}
\expvalue(S+W)\ge \expvalue(S)>0,
\end{align}
which  by \eqref{eq:taun} implies that  $\expvalue(\mathcal{T}_{n}) \rightarrow \infty$ as $n\rightarrow \infty$. As a consequence, by~\eqref{tamoomnemishe2} we also have that $\timepS \rightarrow \infty$ as $n\rightarrow \infty$.

The objective is to design an encoding and decoding strategy,
such that for all $\epsilon, \delta > 0$ and sufficiently large $n$, we have
\begin{equation} \label{objective}
\mathbb{P}(|\XE(\timepS) - \XEH(\timepS)| > \epsilon ) < \delta.
\end{equation}
We have
\begin{align}
&\mathbb{P} ( | \XE(\timepS) - \XEH(\timepS) | > \epsilon ) =  \nonumber \\
&\mathbb{P}(|\XE(\timepS) - \XEH(\timepS)| > \epsilon \mid \timepS  \geq \mathcal{T}_n ) \mathbb{P} (  \timepS  \geq \mathcal{T}_n ) \nonumber \\
&+ \mathbb{P} ( |\XE(\timepS) - \XEH(\timepS) |> \epsilon \mid \timepS < \mathcal{T}_n ) \mathbb{P} (  \timepS < \mathcal{T}_n )\nonumber \\
&\leq \mathbb{P} ( |\XE(\timepS) - \XEH(\timepS) | > \epsilon \mid \timepS \geq \mathcal{T}_n ) +\mathbb{P} (  \timepS < \mathcal{T}_n ),\label{eq:split}
\end{align}
where, using Lemma~\ref{lemnewsub20}, the second term in the sum~\eqref{eq:split},  tends to zero as $n \rightarrow \infty$.
It follows that to ensure~\eqref{objective}  it suffices to design an encoding and decoding scheme,
such that for all $\epsilon, \delta > 0$ and sufficiently large $n$, we have that the following conditional probability is upper bounded by $\delta$.
\begin{align}
\label{jiijeij!1!!22222!!3}
\mathbb{P} ( |\XE(\timepS) - \XEH(\timepS) | > \epsilon \mid \timepS \geq \mathcal{T}_n )<\delta.
\end{align}
From the open-loop equation~\eqref{estprob}, we  have
\begin{align}\label{!!33jifjokod}
\XE(\timepS)= e^{a \timepS}X(0),
\end{align}
from which it follows that the decoder can construct the estimate    \begin{align}\label{wjwidijdej1111}
    \XEH(\timepS)= 
    e^{a\timepS}\hat{X}_{\timepS}(0),
\end{align}
where ${\hat{X}_{\timepS}(0)}$ is an estimate of $X(0)$ constructed at time $\timepS$ using all the symbols received by this time.

By~\eqref{!!33jifjokod} and~\eqref{wjwidijdej1111}, we now have that~\eqref{jiijeij!1!!22222!!3} is equivalent to 
\begin{align} \label{eq:epstrick}
\mathbb{P} ( |X(0) - \hat{X}_{\timepS}(0) | > \epsilon e^{-a\timepS}  \mid \timepS \geq \mathcal{T}_n ) < \delta,
\end{align}
namely it suffices to design an encoding and decoding scheme to communicate the initial condition with exponentially increasing reliability in probability. 
Our  coding procedure that achieves this objective is described next. 
\subsubsection*{Source coding}
We let the source coding map 
\begin{align} \label{scmap}
\mathcal{Q}: [-L,L] \rightarrow \{0,1\}^{\mathbb{N}}
\end{align}
  be an infinite tree-structured quantizer~\cite{gersho2012vector}. This map constructs the infinite binary sequence $\mathcal{Q}\left(X(0)\right)= \{Q_1,Q_2, \ldots\}$ as follows.  $Q_1=0$ if $X(0)$ falls into the left-half of the interval $[-L,L]$, otherwise $Q_1=1$. The sub-interval where $X(0)$ falls is then divided into half and we let  $Q_2=0$ if $X(0)$ falls into the left-half of this sub-interval, otherwise $Q_2=1$. The process then continues in the natural way,  and $Q_i$ is determined accordingly for all $i \ge 3$.

Using the definition of truncation operator ~\eqref{tunc-oper22},
 for any $n' \geq 1$ we can define
\begin{align} \label{qdir}
    \mathcal{Q}_{n'}=\pi_{n'} \circ \mathcal{Q}.
\end{align} 
It follows that   $\mathcal{Q}_{n'}\left(X(0)\right)$ is a binary sequence of length $n'$ that identifies an interval of length $L/2^{n'-1}$ that contains $X(0)$. We also let
 \begin{align} \label{qinv}
\mathcal{Q}_{n'}^{-1}: \{0,1\}^{n'} \rightarrow [-L,L]
\end{align}
  be the  right-inverse map of $\mathcal{Q}_{n'}$, which resolves to the middle point of the last interval identified by the sequence that contains $X(0)$. 
  It follows that for any $n' \geq 1$,  this procedure achieves a quantization error
  \begin{align}\label{quantiz222344!}
    |X(0)-\mathcal{Q}_{n'}^{-1} \circ \mathcal{Q}_{n'} (X(0))| \le \frac{L}{2^{n'}}.
\end{align}
\oprocend

\subsubsection*{Channel coding} 
In order to communicate the quantized initial condition over the timing channel, the truncated binary sequence $Q_{n'}(X(0))$ needs to be mapped into a channel codeword of length $n$. 
  \begin{figure*}[t]
	\begin{subfigure}{0.5\textwidth}
	\centering
   	\includegraphics[scale=0.5]{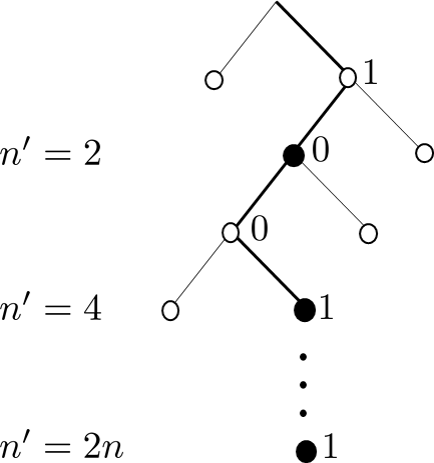}
        \label{fig:Tree11}
  \end{subfigure}~~~~\,\,\,\,\,
  	\begin{subfigure}{0.3\textwidth}
	\centering
   	\includegraphics[scale=0.5]{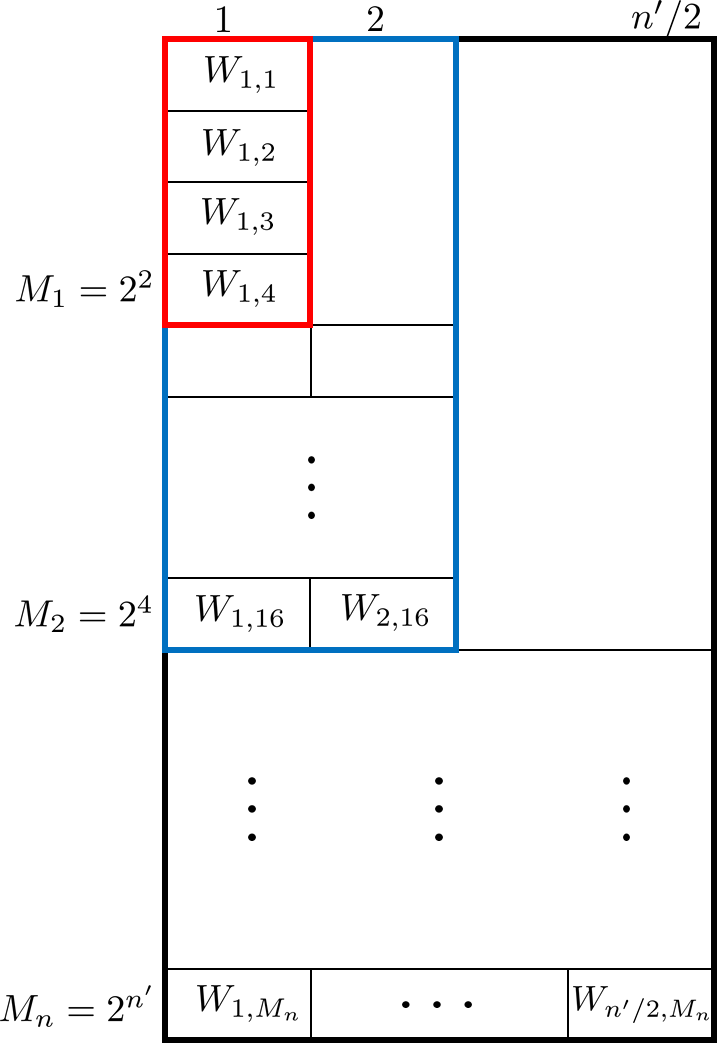}
        \label{fig:tree123}
  \end{subfigure}
     \caption{Tree-structured quantizer and the corresponding codebook for $R \mathbb{E}(D)=2$. In this case, every received channel symbol refines the source coding representation by two bits.  Here the black nodes in the quantization tree at level $n'=\lceil i R \mathbb{E}(D) \rceil=2,4,6,\ldots,$ are mapped into the rows of the codebook. As an example, the first codebook is highlighted in red, and the second in blue.}  
   \label{fig:sns1}
\end{figure*}

\begin{figure*}[t!]
    \begin{subfigure}{0.3\textwidth}
	\centering
     \includegraphics[scale=0.5]{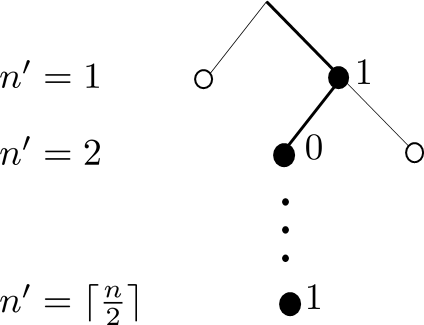}
        \label{fig:codebokk11}
   \end{subfigure}~~~~\,\,\,
    \begin{subfigure}{0.5\textwidth}
	\centering
     	\includegraphics[scale=0.5]{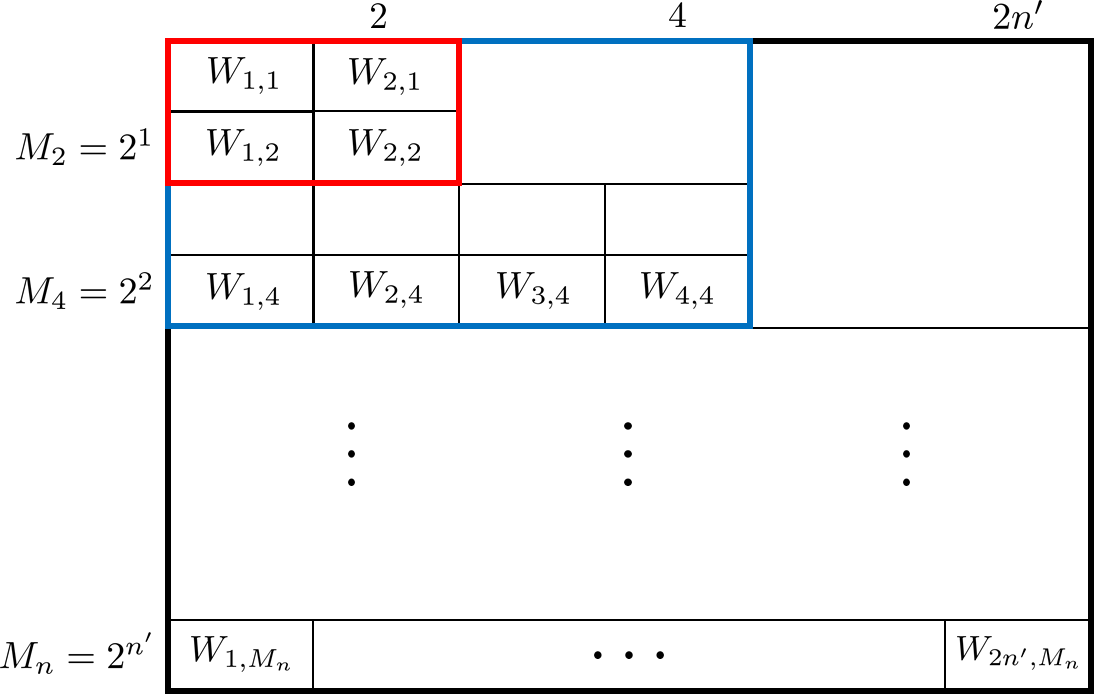}
        \label{fig:codebokk23}
   \end{subfigure}
   \caption{Tree-structured quantizer and the corresponding codebook for $R \mathbb{E}(D)=1/2$. In this case, every two received channel symbols refine the source coding representation by one bit. As an example, the first codebook is highlighted in red, and the second in blue.
  }
   \label{fig:sns}
\end{figure*}
We consider a channel codebook of $n$ columns and $M_n$ rows. The codeword symbols   $ \{ w_{i,m},i=1,\cdots ,n; \, m=1 \cdots M_n  \}$  are   drawn i.i.d. from a distribution which is mixture of a delta function and an exponential and such that  $\mathbb{P}(W_i = 0) = e^{-1}$, and $\mathbb{P}(W_i > w \mid W_i > 0) = \exp\{\frac{-w}{e\mathbb{E}(S)}\}.$   By Theorem 3 of~\cite{anantharam1996bits}, if   the delays $\{S_i\}$ are exponentially distributed, using a maximum likelihood decoder this construction  achieves the timing  capacity. 
Namely, 
letting
\begin{align}\label{nief1268sohrab!!!!}
T_n = \mathbb{E}(\mathcal{T}_n) = n \mathbb{E}(D),
\end{align}
using this codebook we can achieve any rate  
\begin{equation}\label{g1}
 R=\lim_{n\to\infty} \frac{\log M_{n}}{T_{n}} \le C 
\end{equation}
over the timing channel.

Next, we describe the mapping between the source coding and the channel coding constructions.
\oprocend

\subsubsection*{Source-channel mapping}
We first consider the direct mapping. 
For all $i \geq 1$,   we let $n'= \lceil i R \mathbb{E}(D) \rceil$ and consider the $2^{n'}$ possible outcomes of the source coding map $\mathcal{Q}_{n'}(X(0))$. We associate them, in a one-to-one fashion, to the  rows of a codebook  $\Psi_{n'}$ of size $2^{n'} \times  \lceil n'/R \mathbb{E}(D) \rceil$.  
 This mapping is defined as
\begin{align}
\label{edir}
    \mathcal{E}_{n'} : &\{0,1\}^{n'} \rightarrow  \Psi_{n'}.
\end{align}
 By letting $i \rightarrow \infty$, the codebook becomes a double-infinite   matrix $\Psi_{\infty}$, and  the map becomes 
 \begin{align}
 \label{emap}
     \mathcal{E}: \{0,1\}^{\mathbb{N}} \rightarrow \Psi_{\infty}. 
 \end{align} 
Thus, as $i \rightarrow \infty$, $X(0)$ is encoded as
\begin{align}
    X(0) \xrightarrow{\mathcal{Q}} \{0,1\}^{\mathbb{N}}\xrightarrow{\mathcal{E}}  \Psi_\infty. 
\end{align}

We now consider the inverse mapping.
Since the elements of $\Psi_{n'}$ are drawn independently from a continuous distribution, with probability one,  no two rows of the codebook 
are equal to each other, so for any $i \geq 1$ and number of received symbols $n= \lceil i/R \mathbb{E}(D) \rceil$ we   define 
\begin{align} \label{einv}
    \mathcal{E}_{n'}^{-1}: \Psi_{n'} \rightarrow \{0,1\}^{n'},
\end{align}
where $n'=\lceil n R \mathbb{E}(D) \rceil$.
This map associates to every row in the codebook a corresponding node in the quantization tree at level $n'$.

Figures~\ref{fig:sns1} and~\ref{fig:sns} show the constructions described above  for the cases   $R \mathbb{E}(D)=2$ and $R \mathbb{E}(D)=0.5$, respectively.
In Fig.~\ref{fig:sns1}, the nodes in the quantization tree at level $n'=\lceil i R \mathbb{E}(D) \rceil=2,4,6,\ldots,$ are mapped into the rows of a table of $M_n= 2^2,2^4,2^6,
\ldots$ rows and $n=1,2, 3 \ldots $ columns.  Conversely, the  rows in each table  are mapped   into the corresponding nodes in the tree. In Fig.~\ref{fig:sns}, the nodes in the quantization tree at level $n'=\lceil i R \mathbb{E}(D) \rceil=1,2,3,\ldots,$ are mapped into the rows of a table of $M_n=2,2^2,2^3, \ldots$ rows and $n=2,4,6,\ldots$ columns. Conversely, the rows in each table  are mapped into the corresponding nodes in the tree.

Next, we describe how the encoding and decoding operations are performed  using these maps and how transmission occurs over the channel.
\oprocend

\subsubsection*{One-time encoding}
The encoding of the initial state $X(0)$ occurs at the sensor in one-shot and then the corresponding symbols are transmitted over the channel, one by one.   
 Given $X(0)$, the source encoder  computes   $Q (X(0))$ according to the source coding map \eqref{scmap} and the channel encoder picks the corresponding codeword  $\mathcal{E} (Q(X(0)))$ from the doubly-infinite codebook according to the map \eqref{emap}. This codeword is  an infinite sequence of real numbers, which also corresponds to a leaf at infinite depth in the quantization tree.  Then, the encoder starts transmitting the real numbers of the codeword one by one, where each real number corresponds to a holding time, and proceeds in this way forever.  According to the source-channel mapping described above,  transmitting $n =\lceil n'/R \mathbb{E}(D) \rceil$   symbols using this scheme corresponds to transmitting,  for all $i \geq 1$,  $n'= \lceil i R \mathbb{E}(D) \rceil$ source bits, encoded    into
 a codeword $\mathcal{E}_{n'}(Q_{n'}(X(0)))$,  picked from a truncated codebook    of $2^{n'}$ rows and $n$ columns.  
\oprocend

\subsubsection*{Anytime Decoding}
The decoding of the initial state $X(0)$ occurs at the controller in an anytime fashion, refining the estimate of $X(0)$ as more and more symbols are received.

For all $i\geq 1$ the decoder  updates its guess for the value of $X(0)$ each time the number of symbols received equals  $n=\lceil i /R \mathbb{E}(D) \rceil$.  
 Assuming   a decoding operation occurs
 after $n$ symbols have been received, the decoder picks the maximum likelihood codeword from a    truncated codebook of size $M_n \times n$ and by inverse mapping,  it finds the corresponding node in the tree. 
It follows that at the $n$th random reception time $\mathcal{T}_n$, the decoder utilizes the inter-reception times of all $n$    symbols received up to this time
to construct the estimate $\hat{X}_{\mathcal{T}_n}(0)$.  
First, a  maximum likelihood decoder $\mathcal{D}_n$ is employed to map the inter-reception times $(D_1,\dots,D_n)$ to an element of $\Psi_{n'}$. This element is then mapped   to  a binary sequence of length $n'$ using $\mathcal{E}_{n'}^{-1}$. Finally, $\mathcal{Q}_{n'}^{-1}$ is used to construct $\hat{X}_{\mathcal{T}_n}(0)$.
It follows that at the $n$th reception time  where decoding occurs, we have 
\begin{align}\label{succdeco22}
    (D_1,\dots,D_n)& \xrightarrow{\mathcal{D}_n}  
    %\\\nonumber
    \Psi_{n'}
    \xrightarrow{\mathcal{E}_{n'}^{-1}} \{0,1\}^{n'}  
    \xrightarrow{\mathcal{Q}_{n'}^{-1}} [-L,L],
\end{align}
 and we let
 \begin{align}
 \label{forri333!!!23334054}
     \hat{X}_{\mathcal{T}_n}(0)=\mathcal{Q}_{n'}^{-1}\left( \mathcal{E}_{n'}^{-1}\left(\mathcal{D}_n(D_1,\dots,D_n)\right)\right).
 \end{align}
Thus, as $n \rightarrow \infty$ the final decoding process becomes
 \begin{align}\label{succdeco22final}
    (D_1,D_n,\dots)
    %&
    \xrightarrow{\mathcal{D}}  
    %\\\nonumber
    \Psi_\infty
    \xrightarrow{\mathcal{E}^{-1}} \{0,1\}^{\mathbb{N}} %\rightarrow 
    %\nonumber \\
    \xrightarrow{\mathcal{Q}^{-1}} [-L,L].
\end{align}
 \oprocend

To conclude the proof, we now show that if  $C \ge (\gmm) a$, then   it is possible to perform the above encoding and decoding operations with an arbitrarily small probability of error while using a codebook so large that it can accommodate a quantization error at most $L/2^{n'} < \epsilon e^{-a \timepS}$.

Since the channel coding scheme achieves the timing capacity, we have that for any $R \le C$, as $n\rightarrow \infty$ the maximum likelihood decoder  selects  the correct transmitted codeword with arbitrarily high probability.
It follows that for any $\delta>0$ and $n$ sufficiently large, we have   with probability at least $(1-\delta)$  that
\begin{align}
    \mathcal{Q}_{n'} (X(0))= \mathcal{E}_{n'}^{-1}\left(\mathcal{D}_n(D_1,\dots,D_n)\right),
\end{align}
and then by~\eqref{quantiz222344!} we have
\begin{align}\label{eomeo2!!!!!!8765}
    |X(0)-\hat{X}_{\mathcal{T}_n}(0)| \le \frac{L}{2^{n'}}.
\end{align}
We now consider a sequence of estimation times $\{\timepS\}$ satisfying \eqref{tamoomnemishe2} and let the estimate at time $\timepS \geq \mathcal{T}_n$ in \eqref{eq:epstrick} be $\hat{X}_{\timepS}(0)=\hat{X}_{\mathcal{T}_n}(0)$. By \eqref{eomeo2!!!!!!8765} we have that the sufficient condition for estimation reduces to 
\begin{equation} \label{codebookcondition}
\frac{L}{2^{n'}}   \leq \epsilon e^{-a\timepS},
\end{equation}
which means having the size of the codebook $M_n$ be  such that
\begin{equation}
    \frac{L}{M_n}   \leq \epsilon e^{-a\timepS},
\end{equation}
or equivalently
\begin{equation}\label{eq:suffcod}
\frac{\log M_{n} - \log L+ \log \epsilon}{\timepS} \geq a.
\end{equation}
Using \eqref{nief1268sohrab!!!!}, we   have
\begin{align} \label{eq:limit}
\frac{\log M_{n} - \log L + \log \epsilon}{\timepS} &= \frac{\log M_{n} - \log L+ \log \epsilon}{T_{n}} \cdot \frac{T_{n}}{\timepS} \nonumber \\
&=\frac{\log M_{n} - \log L + \log \epsilon}{T_{n}}  \cdot \frac{\mathbb{E}(\mathcal{T}_{n})}{\timepS}.
\end{align}
Taking  the limit for $n \rightarrow \infty$, 
we have 
\begin{align}
\lim_{n\to\infty} \frac{\log M_{n} - \log L + \log \epsilon}{T_{n}} \cdot \frac{\mathbb{E}(\mathcal{T}_{n})}{\timepS}  & \geq  R \cdot \frac{1}{\gmm}. 
\end{align}
It follows that as $n \rightarrow \infty$  the sufficient condition \eqref{eq:suffcod} can be expressed in terms of the rate as
\begin{equation}
R \geq   (\gmm) a.
\end{equation}
It follows that the   rate must satisfy 
\begin{equation}
C \ge R \geq   (\gmm) a 
\end{equation}
and since $C \ge (\gmm) a$, the proof is complete.
\end{proof}

\bibliography{mybib} 
\bibliographystyle{IEEEtran}

\end{document}